\documentclass{article}

\usepackage[letterpaper,top=2cm,bottom=2cm,left=3cm,right=3cm,marginparwidth=1.75cm]{geometry}

\usepackage{amsmath,amssymb,amsthm}
\usepackage{booktabs}
\usepackage{threeparttable}
\usepackage{makecell}
\usepackage{pdflscape}
\usepackage{graphicx}
\usepackage{amsfonts}
\usepackage{multirow}
\usepackage{float}
\usepackage{cite}
\usepackage{listings}
\usepackage[table,svgnames]{xcolor}
\usepackage{changepage}
\usepackage{mathtools}

\usepackage[colorlinks=true, allcolors=blue]{hyperref}

\newcommand{\ms}[2]{\makecell[c]{#1\\[-1pt]{\scriptsize(#2)}}}
\newcommand{\bestms}[2]{\makecell[c]{\textbf{#1}\\[-1pt]{\scriptsize(#2)}}}

\newtheorem{proposition}{Proposition}[section]

\definecolor{codegray}{gray}{0.45}
\definecolor{codebackground}{gray}{0.96}
\definecolor{codekeyword}{RGB}{0,60,120}
\definecolor{codecomment}{RGB}{70,110,70}
\definecolor{codestring}{RGB}{140,40,40}

\lstdefinestyle{Rstyle}{
  language=R,
  alsoletter={.},
  basicstyle=\ttfamily\small,
  keywordstyle=\bfseries\color{codekeyword},
  commentstyle=\itshape\color{codecomment},
  stringstyle=\color{codestring},
  backgroundcolor=\color{codebackground},
  frame=single,
  framerule=0.3pt,
  rulecolor=\color{codegray},
  numbers=left,
  numberstyle=\tiny\color{codegray},
  numbersep=6pt,
  stepnumber=1,
  breaklines=true,
  breakatwhitespace=false,
  breakindent=1em,
  columns=fullflexible,
  keepspaces=true,
  tabsize=2,
  showstringspaces=false,
  captionpos=b,
  aboveskip=0.8\baselineskip,
  belowskip=0.8\baselineskip
}

\title{Priority-preserving augmentation of goodness-of-fit tests by conditional calibration}
\author{Roman Guchenko}

\providecommand{\keywords}[1]
{
  \small	
  \textbf{\textit{Keywords ---}} #1
}

\begin{document}

\maketitle

\begin{abstract}
An omnibus goodness-of-fit statistic can fail to exploit simple
diagnostic evidence efficiently.  For example, under a standard normal
null, an exceptionally large observation is direct evidence of a
scale or tail departure even when the primary omnibus statistic does
not cross its critical value.  This motivates a simple augmentation
principle: retain an established primary test, but reserve a small
part of its null rejection budget for secondary statistics that encode
natural features such as variation, asymmetry, or tail behavior.

We implement this principle by calibrating each secondary acceptance
region under the null conditional on acceptance at all preceding
stages.  Once calibrated, the resulting test is a fixed rectangular
acceptance rule; the ordering is a mechanism for choosing its
boundaries and assigning ordered first-rejection contributions, not a
sequential-sampling scheme.  An unconditional stage-budget
parameterization makes the central trade-off explicit: additional
sensitivity is purchased by removing a prespecified, usually small,
amount of rejection probability from the primary stage.  We establish
strong consistency of the quantile-based Monte Carlo calibration and
give an observation-specific pooled-rank version with exact
randomized finite-\(m\) null size.

In an experiment under a standard normal null with \(n=10\), we
augment the Kolmogorov--Smirnov statistic with sample variance and
sample skewness.  Assigning only \(0.75\%\) of the total Type~I error
budget to the two secondary statistics changes power against
\(\mathcal N(0.5,1)\) from \(0.2742\) to \(0.2733\), while increasing
power against \(\mathcal N(0,0.2^2)\) from \(0.4693\) to \(0.9827\).
This focused experiment is a proof of principle rather than an
exhaustive comparison of normality tests: a small unconditional
allocation to prespecified diagnostics can greatly broaden power while
preserving almost all of the primary test's power.
\end{abstract}


\section{Introduction}

Consider testing a standard normal null.  If one observation in a
small sample is close to \(10\), there is immediate and easily
interpretable evidence against the null.  Nevertheless, a primary
omnibus statistic need not use that evidence efficiently.  For
example, the contribution of a single extreme upper order statistic to
the Kolmogorov--Smirnov distance can be only about \(1/n\).  A simple
scale or tail diagnostic may identify the same departure much more
decisively.

This example reflects a general problem.  The alternative
\(P\neq P_0\) contains changes in location, scale, tail weight,
asymmetry, modality, and many other features.  A statistic designed to
summarize global discrepancy cannot be uniformly sensitive to all of
them.  Numerical comparisons accordingly show substantial changes in
the power ranking of goodness-of-fit tests across alternative
families \cite{rolke2024univariate,rolke2026multivariate}.

The starting point of this paper is therefore deliberately simple.
Rather than replace a useful omnibus test, we augment it with a small
number of natural diagnostics.  The original statistic remains the
\emph{primary} statistic.  Quantities such as sample variance,
skewness, kurtosis, or other problem-specific summaries become
\emph{secondary} statistics.  Only a small part of the total
Type~I error budget is removed from the primary statistic and assigned
to these additions.  Their calibrated acceptance intervals formalize
the natural bounds suggested by the null, such as ``variance should
not be too far from one'' or ``skewness should not be too far from
zero.''  The intended outcome is not uniform improvement: it is a
controlled exchange of a very small amount of primary power for much
greater sensitivity to departures that the primary statistic handles
poorly.

Naively rejecting whenever any added statistic is extreme would
inflate the null rejection probability.  It would also obscure how
much each added statistic contributes after the primary test has been
given priority.  We address both issues through an ordered chain
\[
T_0\longrightarrow T_1\longrightarrow\cdots\longrightarrow T_L.
\]
At stage \(r\), the acceptance region of \(T_r\) is calibrated under
the null conditional on the sample having passed stages
\(0,\ldots,r-1\).  The chain is thus the calibration and accounting
mechanism for the augmentation principle.  It preserves the
prespecified overall size, and its disjoint first-rejection regions
measure what each secondary statistic adds after all earlier
statistics have accepted the sample.  After calibration, however, the
observed sample is tested against one fixed intersection of
coordinatewise acceptance conditions.  Thus ``sequential'' refers to
boundary construction and attribution, not to sequential data
collection or to a requirement that the statistics be evaluated one
at a time.

It is particularly useful to describe the allocation through
unconditional first-rejection budgets \(b_0,\ldots,b_L\), satisfying
\[
\sum_{r=0}^L b_r=\alpha.
\]
The primary budget \(b_0\) is close to \(\alpha\), while the secondary
budgets are normally small.  The corresponding conditional levels are
chosen so that stage \(r\) contributes exactly \(b_r\) to the
population null rejection probability.  This parameterization makes
the scientific question transparent: how much additional power is
obtained from each secondary diagnostic in exchange for the amount
\(\alpha-b_0\) removed from the primary test?

The proposed construction is related to methods that combine evidence
from several tests.  Classical procedures combine marginal
\(p\)-values through functions such as Fisher's product or Tippett's
minimum \cite{fisher1932methods,tippett1931methods}; more recent work
studies \(p\)-value merging under general dependence
\cite{vovkwang2020averaging,vovkwangwang2020admissible,
liuxie2020cauchy}.  Resampling methods instead estimate the joint null
distribution of the component statistics or \(p\)-values
\cite{westfallyoung1993resampling,romanowolf2005stepdown}.  In the
goodness-of-fit setting, Rolke \cite{rolke2020simultaneous} calibrates
the minimum \(p\)-value from a collection of component tests.  Such
simultaneous rules are natural competitors when all components are
treated symmetrically.  Priority and error-budget allocation are also
central in gatekeeping procedures for primary and secondary endpoints
\cite{tamhane2018gatekeeping}, although that literature tests several
hypotheses rather than several statistics for one hypothesis.  Our
objective is to preserve a designated primary test as much as possible
and measure the incremental effect of explicitly chosen secondary
diagnostics.  We use a weighted minimum-\(p\) rule, with weights
matched to the same allocation proportions, as a competing
boundary-selection rule.  The matching is proportional: under
dependence, its jointly calibrated coordinate thresholds do not in
general reproduce the chain's first-rejection budgets.

Other approaches adapt over structured families of components or
tuning parameters.  Neyman's smooth-test framework represents
departures through orthogonal components
\cite{neyman1937smooth,neyman1967selection}, and data-driven
extensions select the number of components or resolution
\cite{ledwina1994datadriven,kallenberg1995consistency,
fan1996wavelet,spokoiny1996adaptive}.  Fromont and Laurent
\cite{fromont2006adaptive} develop an adaptive multiple-testing
procedure based on estimators of an \(L^2\)-distance.  Aggregated
kernel procedures combine tests over kernels or bandwidths
\cite{schrab2022ksdagg,schrab2021mmdagg}.  The secondary statistics
considered here need not form such a parametrized family: they may
encode qualitatively different and directly interpretable features of
the sample.

The ordered construction is also related in spirit to sequentially
rejective multiple testing \cite{goemansolari2010sequential}, but the
setting is different.  We test one null hypothesis using several
statistics of the same sample; we do not test a family of distinct
null hypotheses.  The conditioning is on acceptance by earlier
statistics, not on rejection of earlier hypotheses.

This paper continues our earlier work on simultaneous acceptance
regions in the joint space of several sample statistics
\cite{guchenko2026joint,AldorNoimanBrownBujaRolkeStine2013,
sailynoja2022graphical,king2020}.  The chain replaces one multivariate
calibration problem by successive conditional univariate
calibrations.  When one statistic is added per stage, no multivariate
density estimate is required.

The main contribution is the augmentation principle together with a
priority-preserving calibration scheme that makes its cost and benefit
measurable.  The final test remains a rectangular intersection rule;
the contribution lies in how its boundaries are chosen and how their
incremental effects are accounted for.  We
derive the unconditional and conditional Type~I error
parameterizations and an ordered decomposition of power.  The
supporting theoretical results establish strong consistency of the
reusable Monte Carlo calibration, give a sufficient condition for
population order invariance, and construct an observation-specific
pooled-rank version with exact finite-\(m\) null size.  The weighted
minimum-\(p\)/hyperrectangle equivalence is recorded to clarify the
baseline rather than as a separate competing procedure.

All statistics, tail rules, ordering choices, and budgets in these
validity statements are fixed independently of the observed sample.
If diagnostics or allocations are selected after inspecting the data,
that selection must itself be included in the null calibration.

The experiment augments the Kolmogorov--Smirnov statistic with sample
variance and sample skewness under a standard normal null.  Its
purpose is to evaluate the central claim: a very small unconditional
allocation to well-chosen secondary statistics can greatly broaden
power while preserving almost all of the primary test's power.  A
sensitivity analysis then shows how this exchange changes as the
total secondary budget and its division between variance and skewness
are varied.

\section{Augmenting a primary goodness-of-fit test}

The construction begins with a primary test that one wishes to retain.
Secondary statistics are chosen because they represent simple,
scientifically meaningful departures that may be poorly detected by
the primary statistic.  Their acceptance regions turn qualitative
facts implied by the null into calibrated bounds.  The role of the
chain is to add those bounds while controlling the total null
rejection probability and to quantify the resulting gain and cost.

\subsection{Formal setup}

Let \(X_1,\ldots,X_n\stackrel{\mathrm{iid}}{\sim}P\), where \(P\)
is an unknown absolutely continuous distribution, and write
\(X=(X_1,\ldots,X_n)\in\mathcal X=\mathbb R^n\).  The observed sample
is a realization \(x=(x_1,\ldots,x_n)\).  Throughout the paper,
uppercase letters denote random elements and lowercase letters their
realizations.

We test \(H_0:P=P_0\) against \(H_1:P\neq P_0\), where \(P_0\) is a
fixed null distribution.  Write \(Q_P=P^{\otimes n}\) for the law of
the sample vector and \(Q_0=P_0^{\otimes n}\) under the null.

Now let 
\begin{equation}
\label{eq:stat:seq}
T_0, T_1, \dots, T_L, \qquad T_r : \mathcal{X} \rightarrow \mathbb{R},
\end{equation}
be an ordered sequence of statistics used to assess whether the
observed sample is compatible with \(Q_0\).
We call \(T_0\) the primary statistic and
\(T_1,\ldots,T_L\) the secondary statistics. To stress the importance of order, we write the sequence of statistics~\eqref{eq:stat:seq} as an arrow-connected chain
\begin{equation}
T_0 \to T_1 \to \dots \to T_L,
\end{equation}
and call the test based on this sequence ``the chain test''.

\subsection{Sequential conditional calibration}
The augmented test is defined by the following recursive procedure.

Denote
\begin{equation}
\mathcal{A}_{-1} = \mathcal{X},
\end{equation}
the initial survivor region containing every possible sample.
For stages $r = 0, \dots, L$, let 
\begin{equation}
B_r \subseteq \mathbb{R}
\end{equation}
be the acceptance set (filter bounds) for statistic $T_r$, 
calibrated under the null distribution conditional on
$X\in\mathcal A_{r-1}$.

The corresponding survivor set is defined recursively by
\begin{equation}
\mathcal{A}_{r} = \mathcal{A}_{r-1} \cap \left\{ x \in \mathcal{X} : T_r(x) \in B_r \right\}.
\end{equation}
If \(r<L\), this survivor region defines the conditional null
distribution used to calibrate \(B_{r+1}\).

The augmented test accepts precisely when \(x\in\mathcal A_L\), or
equivalently when
\[
T_r(x)\in B_r
\qquad\text{for every }r=0,\ldots,L.
\]
It rejects otherwise.  The ordering gives priority to the primary
statistic and determines which statistic receives first-rejection
credit; the final rejection region is the complement of the
intersection of all calibrated acceptance conditions.

Consequently, once \(B_0,\ldots,B_L\) have been calibrated, the chain
is applied as the static rectangular rule
\[
\mathcal A_L
=
\bigcap_{r=0}^L\{x:T_r(x)\in B_r\}.
\]
The recursion is therefore a boundary-selection and accounting
device.  It does not enlarge the class of final rejection regions
beyond unions of violations of the chosen coordinatewise bounds.

\subsection{Stage-specific rejection regions}

Let us define the incremental rejection region at stage $r$ as 
\begin{equation}
\mathcal{R}_r = \mathcal{A}_{r-1} \setminus \mathcal{A}_r.
\end{equation}
Equivalently, it can be defined as
\begin{equation}
\mathcal{R}_r = \mathcal{A}_{r-1} \cap \{ x \in \mathcal{X} : T_r(x) \notin B_r \},
\end{equation}
i.e. it consists of all possible samples that pass all the filters up to stage $r-1$ and do not satisfy the stage-$r$ acceptance condition.

Since the regions $\mathcal{R}_r$ are pairwise disjoint by construction, the final rejection region is
\begin{equation}
\label{eq:RL}
\mathcal{R}^{(L)} = \bigcup_{r = 0}^L \mathcal{R}_r = \mathcal{A}^{c}_L.
\end{equation}
Here  
\begin{itemize}
\item $\mathcal{R}_0$ consists of samples rejected by bounds $B_0$ of statistics $T_0$,
\item $\mathcal{R}_1$ consists of samples that pass $B_0$ bounds, but are rejected by $B_1$ bounds of statistic $T_1$,
\item $\mathcal{R}_2$ consists of samples that pass $B_0$ and $B_1$, but are rejected by $B_2$,
\item and so on.
\end{itemize}

\subsection{Conditional levels and unconditional stage budgets}

For \(r=0,\ldots,L\), provided that
\(Q_0(\mathcal A_{r-1})>0\), define
\begin{align}
\label{eq:bounds:eq1}
\gamma_r
=
Q_0\!\left(
    X\in\mathcal R_r
    \mid
    X\in\mathcal A_{r-1}
\right)
=
Q_0\!\left(
    T_r(X)\notin B_r
    \mid
    X\in\mathcal A_{r-1}
\right)
=
\frac{
    Q_0\!\left(
        X\in\mathcal A_{r-1},
        \;
        T_r(X)\notin B_r
    \right)
}{
    Q_0\!\left(
        X\in\mathcal A_{r-1}
    \right)
},
\end{align}
the conditional probability of rejection at stage \(r\), given
acceptance at all preceding stages.  For \(r\geq1\), \(\gamma_r\) is
the conditional level assigned to a secondary statistic.

For the primary stage,
\begin{equation}
\label{eq:bounds:eq2}
\gamma_0 = Q_0 (\mathcal{R}_0),
\end{equation}
because \(\mathcal A_{-1}=\mathcal X\).

The stage-\(r\) acceptance region \(B_r\) can be constructed from
conditional quantiles or another prespecified acceptance-region rule
under the corresponding conditional null distribution.

\medskip

For the augmentation problem, the unconditional probability assigned
to each first-rejection region is more directly interpretable.  Define
\begin{equation}
\label{eq:unconditional-stage-budget}
b_r
:=
Q_0(\mathcal R_r)
=
\left\{\prod_{s=0}^{r-1}(1-\gamma_s)\right\}\gamma_r,
\qquad r=0,\ldots,L,
\end{equation}
where the empty product equals one.  We call \(b_r\) the
unconditional stage-\(r\) budget.  The primary test receives \(b_0\);
the total amount removed from it and assigned to the secondary
statistics is
\[
\sum_{r=1}^L b_r=\alpha-b_0.
\]

Because the first-rejection regions are disjoint,
\begin{align}
Q_0(\mathcal R^{(L)})
&=
\sum_{r=0}^L b_r \nonumber\\
&=
\gamma_0+
\sum_{r=1}^L
\left\{\prod_{s=0}^{r-1}(1-\gamma_s)\right\}\gamma_r.
\label{eq:type1-budget-decomposition}
\end{align}
Fixing the overall level at
\(\alpha=Q_0(\mathcal R^{(L)})\) therefore requires
\[
\sum_{r=0}^L b_r=\alpha.
\]
Equivalently,
\begin{equation}
\label{eq:alpha:decomposition:alternative}
\alpha = {Q}_0(\mathcal{R}^{(L)}) = 1 - \prod_{r = 0}^L (1 - \gamma_r),
\end{equation}
as~\eqref{eq:RL} and
\begin{equation}
{Q}_0(\mathcal{A}_r) = \prod_{s=0}^r (1-\gamma_s).
\end{equation}

Conversely, any nonnegative budgets \(b_0,\ldots,b_L\) with
\(\sum_{r=0}^L b_r=\alpha<1\) determine the conditional levels
\begin{equation}
\label{eq:budget-to-gamma}
\gamma_r
=
\frac{b_r}
{1-\sum_{s=0}^{r-1}b_s},
\qquad r=0,\ldots,L.
\end{equation}
Thus the substantive design choice can be made on the unconditional
scale: choose how much of the total error budget to retain for the
primary statistic and how to divide the remainder among the natural
secondary diagnostics.  Conditional calibration then implements that
choice without exceeding the prescribed population size.

\subsection{Measuring what the secondary statistics add}

The budget decomposition measures the null cost of each stage.  Its
power analogue measures the benefit obtained under a specified
alternative.

Let ${P}_1$ be some alternative distribution.
Define
\begin{align}
{Q}_1 = {Q}_{P_1},
\end{align}
the distribution of samples of size $n$ from ${P}_1$.

Then the power of the chain test against the alternative ${P}_1$ is
\begin{equation}
\pi({P}_1)
=
{Q}_1 \bigl(\mathcal{R}^{(L)}\bigr).
\end{equation}
Since the incremental rejection regions are disjoint,
\begin{equation}
\label{eq:power:decomposition}
\pi({P}_1)
=
\sum_{r=0}^{L} {Q}_1(\mathcal{R}_r)
=
\sum_{r=0}^{L} \Delta_r (P_1)
.
\end{equation}
The quantity
\[
\Delta_r(P_1)=Q_1(\mathcal R_r)
\]
is the probability that a sample generated under \(P_1\) is first
rejected at stage \(r\). We refer to it as the ordered first-rejection
contribution of stage \(r\). In general, it depends on the position of
the statistic in the chain and should not be interpreted as an
order-free measure of the statistic's intrinsic importance.

For a proposed augmentation, the pair
\[
\bigl(b_r,\Delta_r(P_1)\bigr)
\]
summarizes the stage-\(r\) exchange: \(b_r\) is the unconditional null
budget spent on that stage, while \(\Delta_r(P_1)\) is the resulting
first-rejection probability under the alternative.  Comparing these
quantities across alternatives is the main purpose of the chain
decomposition.

\subsection{A secondary question: order invariance}
\label{sec:order-invariance}

The augmentation principle does not require order invariance: the
primary statistic is intentionally placed first, and secondary
statistics may also be ordered by design.  Nevertheless, it is useful
to know when changing the order of the secondary diagnostics leaves
the final test unchanged.

The final rejection region generally depends on the order because the
null distribution used to calibrate a later statistic is conditional
on the acceptance regions of all earlier statistics.  The following
result gives a sufficient condition under which this dependence
disappears at the population level.

Hold the primary acceptance region \(B_0\), and hence
\(\mathcal A_0\), fixed across orderings.  For each secondary statistic
\(T_j\), \(j=1,\ldots,L\), let
\(\eta_j\in[0,1)\) be a conditional rejection level attached to that
statistic, and let \(B_j\) be the unique acceptance region produced by
the chosen population calibration rule from
\begin{equation}
\label{eq:secondary-given-primary}
Q_0\!\left(T_j(X)\notin B_j\mid X\in\mathcal A_0\right)
=
\eta_j.
\end{equation}
When the statistics are reordered, \(\eta_j\) is carried with
\(T_j\); it is not attached to the numerical position of that
statistic in the chain.

\begin{proposition}[Order invariance under conditional null independence]
\label{prop:order-invariance}
Suppose that \(Q_0(\mathcal A_0)>0\) and that
\(T_1(X),\ldots,T_L(X)\) are mutually conditionally independent under
\(Q_0\) given \(X\in\mathcal A_0\).  Equivalently, for all Borel sets
\(C_1,\ldots,C_L\),
\begin{equation}
\label{eq:conditional-null-independence}
Q_0\!\left(
\bigcap_{j=1}^L\{T_j(X)\in C_j\}
\,\middle|\,
X\in\mathcal A_0
\right)
=
\prod_{j=1}^L
Q_0\!\left(
T_j(X)\in C_j
\,\middle|\,
X\in\mathcal A_0
\right).
\end{equation}
Assume also that the statistic-specific levels \(\eta_j\) and the
calibration rule defining each \(B_j\) are unchanged when the
secondary statistics are permuted.

Then, for every permutation \(\sigma\) of
\(\{1,\ldots,L\}\), sequential population calibration of the chain
\[
T_0\longrightarrow T_{\sigma(1)}\longrightarrow\cdots
\longrightarrow T_{\sigma(L)}
\]
uses the same statistic-specific acceptance regions
\(B_1,\ldots,B_L\).  Its final acceptance and rejection regions are
\begin{equation}
\label{eq:order-invariant-regions}
\mathcal A_L^{\sigma}
=
\mathcal A_0
\cap
\bigcap_{j=1}^L\{x:T_j(x)\in B_j\},
\qquad
\mathcal R^{\sigma}
=
(\mathcal A_L^{\sigma})^c,
\end{equation}
and therefore do not depend on \(\sigma\).  Consequently, for every
distribution \(Q\) on \(\mathcal X\), not only for \(Q_0\),
\begin{equation}
\label{eq:order-invariant-power}
Q(\mathcal R^{\sigma})
=
Q(\mathcal R^{\tau})
\end{equation}
for all permutations \(\sigma\) and \(\tau\).  In particular, the
null rejection probability is
\begin{equation}
\label{eq:order-invariant-size}
Q_0(\mathcal R^{\sigma})
=
1-(1-\gamma_0)\prod_{j=1}^L(1-\eta_j).
\end{equation}
The ordered first-rejection contributions need not be invariant.
\end{proposition}

\paragraph{Proof idea.}
Conditional independence implies that, after primary acceptance,
conditioning on acceptance by any subset of the other secondary
statistics does not change the null distribution of \(T_j\).
Consequently, the uniquely calibrated region \(B_j\) is the same at
every position.  The final survivor region is therefore the
order-free intersection in~\eqref{eq:order-invariant-regions};
conditional independence also factors its null probability, giving
\eqref{eq:order-invariant-size}.  The full proof is given in
Appendix~\ref{appendix_proof_order_invariance}.

The requirement that \(\eta_j\) remain attached to \(T_j\) is
substantive.  It is automatically satisfied when all secondary
conditional levels are equal.  By contrast, when fixed unconditional
first-rejection budgets are converted into conditional levels, a
statistic's conditional level may change with its position.  In that
case Proposition~\ref{prop:order-invariance} does not assert exact
invariance, although it suggests approximate invariance when the
resulting conditional levels and boundaries are close.

Under the assumptions of
Proposition~\ref{prop:calibration-consistency}, the corresponding
Monte Carlo boundaries converge almost surely to the common
population boundaries in~\eqref{eq:order-invariant-regions}.  Hence,
for any \(Q\) assigning zero probability to those boundaries and any
two permutations \(\sigma,\tau\),
\[
Q\!\left(
\widehat{\mathcal R}^{\sigma}_m
\mathbin{\triangle}
\widehat{\mathcal R}^{\tau}_m
\right)
\longrightarrow0,
\]
where \(\triangle\) denotes symmetric difference, so finite-calibration order differences vanish asymptotically under
the proposition's conditions.

\subsection{Matched baseline: weighted minimum-\texorpdfstring{\(p\)}{p}
and hyperrectangle tests}
\label{sec:minp-hyperrectangle}

The experiment compares the chain's boundary-selection rule with a
weighted minimum-\(p\) rule whose weights match the same allocation
proportions.  This is a useful priority-weighted comparator, but it
does not isolate ``sequential application'' from ``simultaneous
application'': after calibration, both procedures are fixed
coordinatewise intersection rules.  They differ in how their
coordinate thresholds are selected.  The following elementary
equivalence shows that the weighted minimum-\(p\) baseline is
precisely a hyperrectangle test in component \(p\)-value space.

For \(j=1,\ldots,K\), let \(p_j:\mathcal X\to[0,1]\) be a component
\(p\)-value.  All components are computed from the same observed
sample.  Let \(w_j\geq0\) satisfy \(\sum_{j=1}^K w_j=1\), and define
the active coordinate set \(J_w=\{j:w_j>0\}\).  The weighted
minimum-\(p\) statistic is
\begin{equation}
\label{eq:minp-general}
M_w(x)
=
\min_{j\in J_w}\frac{p_j(x)}{w_j}.
\end{equation}

\begin{proposition}[Equivalence of weighted minimum-\(p\) and hyperrectangle tests]
\label{prop:minp-hyperrectangle}
Fix \(c\geq0\) such that \(cw_j\leq1\) for every \(j\in J_w\), and
use the same weak or strict inequality convention throughout.  Define
the hyperrectangle acceptance region in component \(p\)-value space by
\begin{equation}
\label{eq:pvalue-hyperrectangle}
\mathcal H_w(c)
=
\left\{
p\in[0,1]^K:
p_j>cw_j\text{ for every }j\in J_w
\right\},
\end{equation}
with zero-weight coordinates unrestricted.  Then, pointwise for every
sample \(x\),
\begin{align}
M_w(x)\leq c
&\quad\Longleftrightarrow\quad
\exists j\in J_w:\ p_j(x)\leq cw_j
\nonumber\\
&\quad\Longleftrightarrow\quad
\bigl(p_1(x),\ldots,p_K(x)\bigr)
\notin\mathcal H_w(c).
\label{eq:minp-hyperrectangle-equivalence}
\end{align}
Thus the weighted minimum-\(p\) test and the hyperrectangle test with
coordinate thresholds \(u_j=cw_j\) have identical rejection regions.
They consequently have identical size and power under every
distribution of the data, without any independence assumption.

Conversely, let \(u_j\in[0,1]\), not all zero, define
\(J_u=\{j:u_j>0\}\), and consider the lower-threshold hyperrectangle
that rejects when \(p_j(x)\leq u_j\) for at least one
\(j\in J_u\), with the remaining coordinates unrestricted.  Setting
\begin{equation}
\label{eq:hyperrectangle-to-minp}
c=\sum_{j=1}^K u_j,
\qquad
w_j=\frac{u_j}{c}
\end{equation}
represents the same test as a weighted minimum-\(p\) test.  Hence the
two descriptions parameterize the same class of lower-threshold
hyperrectangular rejection rules in component \(p\)-value space.
\end{proposition}

\paragraph{Proof idea.}
The equivalence follows pointwise by expanding the minimum:
\[
\min_{j\in J_w}\frac{p_j(x)}{w_j}\leq c
\quad\Longleftrightarrow\quad
\exists j\in J_w:\ p_j(x)\leq cw_j.
\]
Conversely, any lower-threshold hyperrectangle with thresholds
\(u_j\) is recovered by taking
\(c=\sum_j u_j\) and \(w_j=u_j/c\).  The full proof is given in
Appendix~\ref{appendix_proof_minp}.

The equivalence is algebraic and therefore continues to hold when the
common threshold is estimated by simulation.  For every realized
Monte Carlo threshold \(\widehat c\), both descriptions make the same
decision on every evaluation sample, provided that they use the same
component \(p\)-values and the same convention for boundary ties.
Thus a jointly calibrated weighted minimum-\(p\) test should not also
be counted as a distinct hyperrectangle baseline.  The result does
not say that the particular weighted minimum-\(p\) baseline and the
chain have the same boundaries.  The chain selects later boundaries
from conditional survivor distributions, whereas the baseline jointly
calibrates a common multiplier applied to prespecified marginal
weights.  Matching the weight proportions therefore does not generally
match the chain's coordinate thresholds or its unconditional
first-rejection budgets.

If each set \(\{x:p_j(x)>u_j\}\) is the pullback of a marginal
statistic acceptance interval \(B_j(u_j)\), then the pullback of
\(\mathcal H_w(c)\) to statistic space is the corresponding Cartesian
product of those intervals.  This includes one-sided component tests
and the usual two-sided component \(p\)-values.  Conversely, the
calibrated chain is itself a hyperrectangle in statistic space.  The
empirical comparison below should therefore be read as a comparison
of two boundary-selection criteria within this broad rectangular
class, not as evidence that evaluating the same fixed bounds
sequentially would change the decision.

\subsection{Design choices and prespecification}
\label{sec:design-prespecification}

The size calculations above treat the statistics, their tail rules,
their order, and the budgets \(b_0,\ldots,b_L\) as fixed parts of the
test design.  In particular, the observed sample must not be inspected
to decide which secondary statistic to add or how much budget to give
it.  Such post hoc selection changes the rejection rule and can
invalidate the stated null level unless the complete selection
procedure is reproduced inside the null calibration.

A practical design may instead use scientific knowledge and
independent design simulations.  One can prespecify departure classes
that matter in the application, choose interpretable diagnostics for
those classes, and select a small secondary budget using a stated
utility, prior weighting, or worst-case power criterion.  If the
design simulations are used to tune the allocation, their alternative
grid and objective should be reported, and performance should be
checked on additional alternatives.  The null calibration is then
applied only after the complete rule has been fixed.

The same requirement applies to the ordering.  The primary statistic
is placed first by scientific priority.  Secondary order can be chosen
by interpretability or by an externally specified priority, but the
ordered contributions should not be searched post hoc for the most
favorable attribution.

\subsection{Possible generalizations}

The same principle may be extended to stages containing several
statistics by calibrating a joint stagewise acceptance region under
the corresponding conditional null distribution. A two-sample
version may likewise be constructed using permutation calibration
under an appropriate exchangeability null. We leave a detailed study
of these extensions to future work.

\section{Monte Carlo calibration and supporting guarantees}

The augmentation principle can be implemented with one reusable bank
of null simulations.  The results in this section justify that
implementation and then describe a separate observation-specific
construction when exact finite-\(m\) null size is required.  These
results support the main power-allocation argument; they do not change
the substantive choice of primary and secondary statistics.

\subsection{Reusable calibration of the stage boundaries}

Let
\begin{equation}
\label{eq:sample:of:samples:null}
X_1^*,\ldots,X_m^*
\stackrel{\mathrm{iid}}{\sim}
Q_0
\end{equation}
be independent null sample vectors, each representing a sample of size
\(n\).

Let \(I_{r,m}\) be the indices of calibration vectors that pass every
filter through stage \(r\).

We start with
\begin{equation}
I_{-1,m}
=
\{1,\ldots,m\}.
\end{equation}
At each stage \(r=0,\ldots,L\), we construct the acceptance region
\(\widehat B_{r,m}\) so that
\begin{equation}
\label{eq:bounds:simulation}
\widehat{\gamma}_{r,m}
:=
\frac{
    \#\left\{
        i\in I_{r-1,m}:
        T_r(X_i^*)\notin\widehat B_{r,m}
    \right\}
}{
    |I_{r-1,m}|
}
\approx
\gamma_r
\end{equation}
provided that \( |I_{r-1,m}|>0 \).

After that, we define
\begin{equation}
I_{r,m}
=
\left\{
i\in I_{r-1,m}
:
T_r(X_i^*)
\in
\widehat{B}_{r,m}
\right\}.
\end{equation}

\subsection{Strong consistency of sequential Monte Carlo calibration}
\label{sec:calibration-consistency}

We now justify the calibration procedure when the acceptance regions
are constructed from conditional empirical quantiles.  Define the
empirical measure of the calibration sample by
\[
\mathbb Q_m(C)
=
\frac{1}{m}\sum_{i=1}^m \mathbf 1\{X_i^*\in C\},
\qquad C\subseteq\mathcal X,
\]
and define the random survivor regions
\[
\widehat{\mathcal A}_{-1,m}=\mathcal X,
\qquad
\widehat{\mathcal A}_{r,m}
=
\widehat{\mathcal A}_{r-1,m}
\cap
\{x:T_r(x)\in\widehat B_{r,m}\}.
\]
Then
\(I_{r,m}=\{i:X_i^*\in\widehat{\mathcal A}_{r,m}\}\) and hence
\[
\frac{|I_{r,m}|}{m}
=
\mathbb Q_m(\widehat{\mathcal A}_{r,m}).
\]
For the population construction, let
\[
F_r(t)
=
Q_0\!\left(T_r(X)\le t\mid X\in\mathcal A_{r-1}\right)
\]
whenever \(Q_0(\mathcal A_{r-1})>0\).

\begin{proposition}[Strong consistency of sequential Monte Carlo calibration]
\label{prop:calibration-consistency}
Suppose that \(L<\infty\) and that, at every active stage \(r\), the
population acceptance interval \(B_r\) is determined by a fixed finite
collection of quantiles of \(F_r\).  Construct
\(\widehat B_{r,m}\) from the corresponding quantiles of the empirical
conditional distribution based on \(I_{r-1,m}\).  Assume that:
\begin{enumerate}
\item \(Q_0(\mathcal A_{r-1})>0\) for every \(r=0,\ldots,L\);
\item each quantile used to define \(B_r\) is unique, and \(F_r\) is
continuous and strictly increasing in a neighborhood of that
quantile;
\item an inactive stage has \(\gamma_r=0\) and is represented by
\(B_r=\widehat B_{r,m}=\mathbb R\).
\end{enumerate}
Then, with \(Q_0^\infty\)-probability one, simultaneously for all
\(r=0,\ldots,L\) as \(m\to\infty\):
\begin{enumerate}
\item every finite endpoint of \(\widehat B_{r,m}\) converges to the
corresponding endpoint of \(B_r\);
\item
\begin{equation}
\label{eq:survivor-consistency}
\frac{|I_{r,m}|}{m}
\longrightarrow
Q_0(\mathcal A_r)
=
\prod_{s=0}^r(1-\gamma_s);
\end{equation}
\item the empirical conditional rejection proportion in
\eqref{eq:bounds:simulation} satisfies
\begin{equation}
\label{eq:gammahat-consistency}
\widehat\gamma_{r,m}\longrightarrow\gamma_r;
\end{equation}
\item the population conditional rejection probability of the
Monte Carlo region satisfies
\begin{equation}
\label{eq:gammatilde-consistency}
\widetilde\gamma_{r,m}
:=
Q_0\!\left(
T_r(X)\notin\widehat B_{r,m}
\mid
X\in\widehat{\mathcal A}_{r-1,m}
\right)
\longrightarrow
\gamma_r.
\end{equation}
\end{enumerate}
Consequently, if
\(\widehat{\mathcal R}^{(L)}_m
=\widehat{\mathcal A}_{L,m}^{c}\), then
\begin{equation}
\label{eq:size-consistency}
Q_0\!\left(\widehat{\mathcal R}^{(L)}_m\right)
\longrightarrow
1-\prod_{r=0}^L(1-\gamma_r).
\end{equation}
In particular, when the conditional levels satisfy
\eqref{eq:alpha:decomposition:alternative}, the calibrated chain has
asymptotic null rejection probability \(\alpha\).
\end{proposition}

\paragraph{Proof idea.}
The vector of statistics induces a finite-dimensional class of
axis-aligned rectangles, which is Glivenko--Cantelli.  Uniform
convergence over this class controls the random survivor regions.
An induction over the stages then gives uniform convergence of each
conditional empirical distribution function, consistency of its
quantile endpoints, and convergence of the survivor probabilities.
The conditional rejection probabilities and the final null size
follow by taking ratios and complements.  The full proof is given in
Appendix~\ref{appendix_proof_consistency}.

Proposition~\ref{prop:calibration-consistency} concerns randomness in
the calibration regions.  It also yields consistency of the
independent evaluation step.  Conditional on the calibration sample,
\(\widehat\pi_{m,h}(P_1)\) is an average of \(h\) Bernoulli variables
with mean \(Q_1(\widehat{\mathcal R}^{(L)}_m)\), and therefore
\[
\Pr\!\left(
\left|
\widehat\pi_{m,h}(P_1)
-Q_1(\widehat{\mathcal R}^{(L)}_m)
\right|>\varepsilon
\,\middle|\,
X_1^*,\ldots,X_m^*
\right)
\le
\frac{1}{4h\varepsilon^2}.
\]
In the null case \(P_1=P_0\), it follows that, as \(m,h\to\infty\),
\(\widehat\alpha_{m,h}\to\alpha\) in probability.  More generally, if
\(Q_1\) assigns zero probability to the population boundaries, then
\(\widehat\pi_{m,h}(P_1)\to\pi(P_1)\) in probability.

\subsection{An optional exact pooled-rank version}
\label{sec:exact-pooled-rank}

The preceding construction estimates acceptance regions from a null
calibration sample and then applies those regions to the observed
sample.  Its null rejection probability is asymptotically correct by
Proposition~\ref{prop:calibration-consistency}, but exact equality at a
fixed \(m\) need not be attainable because empirical quantiles lie on
a discrete grid.  We now describe an observation-specific alternative
that obtains exact finite-\(m\) null level by pooling the observed
sample with the simulated null samples and treating all pooled rows
symmetrically.

Use the same unconditional first-rejection budgets \(b_r\) defined
in~\eqref{eq:unconditional-stage-budget}.  Thus
\begin{equation}
\label{eq:pooled-budget-sum}
\sum_{r=0}^L b_r=\alpha.
\end{equation}

For the theoretical null analysis, let
\[
X_0^*,X_1^*,\ldots,X_m^*
\stackrel{\mathrm{iid}}{\sim}Q_0,
\qquad M=m+1.
\]
In an application, \(X_0^*\) is replaced by the observed realization
\(x\), while \(X_1^*,\ldots,X_m^*\) are generated from \(Q_0\).  Let
\(U\sim\operatorname{Uniform}[0,1)\), independently of all pooled
samples, and define
\begin{equation}
\label{eq:pooled-cumulative-counts}
C_{-1,M}=0,
\qquad
C_{r,M}
=
M\sum_{s=0}^r b_s,
\end{equation}
and
\begin{equation}
\label{eq:pooled-stage-counts}
K_{r,M}
=
\left\lfloor C_{r,M}+U\right\rfloor
-
\left\lfloor C_{r-1,M}+U\right\rfloor.
\end{equation}
The integer \(K_{r,M}\) is the number of pooled rows assigned to first
rejection at stage \(r\).  Since
\[
E\{\lfloor c+U\rfloor\}=c
\]
for every \(c\geq0\),
\begin{equation}
\label{eq:pooled-count-expectation}
E(K_{r,M})=Mb_r.
\end{equation}
Moreover,
\begin{equation}
\label{eq:pooled-total-count}
K_M
:=
\sum_{r=0}^L K_{r,M}
=
\lfloor M\alpha+U\rfloor,
\qquad
E(K_M)=M\alpha.
\end{equation}

Set
\[
\widetilde I_{-1,M}=\{0,1,\ldots,m\}.
\]
At stage \(r\), among the indices in
\(\widetilde I_{r-1,M}\), apply the prespecified tail rule for
\(T_r\) and select exactly \(K_{r,M}\) indices having the most extreme
values of \(T_r(X_i^*)\).  Denote the selected set by
\[
\widetilde D_{r,M}
\subseteq
\widetilde I_{r-1,M},
\qquad
|\widetilde D_{r,M}|=K_{r,M},
\]
and update
\begin{equation}
\label{eq:pooled-survivor-recursion}
\widetilde I_{r,M}
=
\widetilde I_{r-1,M}\setminus\widetilde D_{r,M}.
\end{equation}
For a right-tailed statistic, the \(K_{r,M}\) largest surviving values
are selected.  For an equal-tailed statistic, the count is divided
between the smallest and largest surviving values; if the count is
odd, an independent fair randomization determines which tail receives
the additional row.  Ties are resolved using iid continuous random
priorities attached to the pooled rows.  More generally, all that is
required below is that the stagewise selection rule be
permutation-equivariant and select exactly \(K_{r,M}\) surviving
indices.

The pooled-rank chain rejects the observed sample if and only if
\begin{equation}
\label{eq:pooled-rejection-decision}
0\in\bigcup_{r=0}^L\widetilde D_{r,M}.
\end{equation}
If index zero first belongs to \(\widetilde D_{r,M}\), the observed
sample is first rejected at stage \(r\).
In the following result, \(\Pr_0\) denotes joint probability over the
pooled null samples and all auxiliary randomization.

\begin{proposition}[Exact finite-\texorpdfstring{\(m\)}{m} validity of the pooled-rank chain]
\label{prop:pooled-rank-exact}
Suppose that \(Q_0\) is fixed and that, under \(H_0\),
\(X_0^*,\ldots,X_m^*\) are iid from \(Q_0\).  Suppose also that all
rounding, tail-allocation, and tie-breaking variables are independent
of the pooled samples, and that every stagewise selection rule is
permutation-equivariant.  Then, including the auxiliary
randomization,
\begin{equation}
\label{eq:pooled-exact-stage}
\Pr_0\!\left(0\in\widetilde D_{r,M}\right)=b_r,
\qquad r=0,\ldots,L,
\end{equation}
and
\begin{equation}
\label{eq:pooled-exact-size}
\Pr_0\!\left(
0\in\bigcup_{r=0}^L\widetilde D_{r,M}
\right)
=
\alpha.
\end{equation}
Furthermore, whenever
\(\Pr_0(0\in\widetilde I_{r-1,M})>0\),
\begin{equation}
\label{eq:pooled-exact-conditional}
\Pr_0\!\left(
0\in\widetilde D_{r,M}
\mid
0\in\widetilde I_{r-1,M}
\right)
=
\gamma_r.
\end{equation}
\end{proposition}

\paragraph{Proof idea.}
Conditional on the randomized integer stage counts, exchangeability
of the pooled rows and permutation equivariance give every row the
same probability \(K_{r,M}/M\) of receiving the stage-\(r\)
first-rejection label.  Averaging over the systematic rounding yields
the unconditional budget \(b_r\).  The stagewise rejection sets are
disjoint, so their probabilities sum to \(\alpha\); division by the
survival probability through stage \(r-1\) gives \(\gamma_r\).
The full proof is given in
Appendix~\ref{appendix_proof_pooled_rank}.

The exactness in Proposition~\ref{prop:pooled-rank-exact} includes the
auxiliary randomization.  If \(M\alpha\) is an integer, then
\(K_M=M\alpha\) for every value of \(U\), so the total pooled rank
level is already exactly \(\alpha\).  If randomization is undesirable,
one may instead use \(K_M=\lfloor M\alpha\rfloor\); the resulting
nonrandomized test has null rejection probability
\(\lfloor M\alpha\rfloor/M\leq\alpha\), with the stage allocations
rounded to the finite rank grid.

For continuous statistics, selecting prescribed lower and upper ranks
is equivalent to choosing order-statistic boundaries.  Thus the
pooled procedure does select stagewise boundaries indirectly.  Unlike
\(\widehat B_{r,m}\), however, these implicit boundaries depend on the
observed sample because \(X_0^*=x\) participates in every pooled
ranking.  They are therefore observation-specific rather than one set
of reusable calibration regions.

This distinction explains why the pooled-rank construction is a
finite-\(m\) alternative rather than a replacement for the procedure
implemented below.  The null-only calibration method constructs
\(\widehat B_{0,m},\ldots,\widehat B_{L,m}\) once and can apply them to
many observations or to a large power-evaluation sample.  The
pooled-rank method gives an exact Monte Carlo decision for one observed
sample, but a new pooled ranking is required for each observation.
Accordingly, the experiment in Section~4 studies the reusable
null-calibrated procedure.  No equality of finite-\(m\) power between
the two constructions is asserted.  Exactness also requires a fixed
null law \(Q_0\), or another simulation scheme that preserves pooled
exchangeability; it does not automatically extend to a composite null
with parameters estimated from the observed sample.

\subsection{Estimating power and Type I error}

To evaluate power, we simulate
\begin{equation}
\label{eq:sample:of:samples:alternative}
Y_1^*,\ldots,Y_h^*
\stackrel{\mathrm{iid}}{\sim}
Q_1
\end{equation}
and then take 
\begin{equation}
\label{eq:estimate:power}
\widehat{\pi}_{m,h} ({P}_1) = \frac{ 
\#\left\{
        i\in\{1,\ldots,h\}:
        \exists r\in\{0,\ldots,L\}
        \text{ such that }
        T_r(Y_i^*)\notin\widehat{B}_{r,m}
    \right\}
}{h} 
\end{equation}
as a power estimate. 

Type I error estimate then is
\begin{equation}
\widehat{\alpha}_{m,h} = \widehat{\pi}_{m,h} ({P}_0),
\end{equation}
i.e. we use the same formula as for power estimate, but simulate $Y_1^*, \dots, Y_h^*$ from ${Q}_0$ instead of ${Q}_1$.

Both \(\widehat{\alpha}_{m,h}\) and
\(\widehat{\pi}_{m,h}(P_1)\) are random variables. Their variability arises from
the calibration sample \(X_1^*,\ldots,X_m^*\), through the estimated
regions \(\widehat B_{r,m}\), and from the independent evaluation
sample \(Y_1^*,\ldots,Y_h^*\).
Therefore, it is desirable to report not a single power estimate but
the distribution induced by the complete calibration/evaluation
procedure.  In the experiments below, we repeat
\eqref{eq:sample:of:samples:null} and
\eqref{eq:sample:of:samples:alternative} \(t\) times.  We report the
empirical mean and its Monte Carlo standard error.  The corresponding
repetition-to-repetition standard deviation is the standard error
multiplied by \(\sqrt t\); it describes the variability of one
complete finite-\(m\), finite-\(h\) run rather than the precision of
the reported average.

Because procedures are evaluated on the same samples within each
repetition, their differences are paired.  For comparisons of primary
interest we therefore also report the standard error of the
within-repetition power difference; marginal standard errors alone do
not measure uncertainty in a paired comparison.

The same repetitions can be used to quantify
calibration-to-calibration variation of every estimated boundary.
Appendix
\ref{appendix_boundary_uncertainty} separates the order-statistic
uncertainty at a fixed stage from the additional uncertainty propagated
through the preceding estimated filters.

In addition, we estimate the ordered first-rejection contributions
$\Delta_r(P_1)$ by
\begin{equation}
\widehat{\Delta}_{r;m,h} ({P}_1) = 
\frac{\# \left\{ i \in \left\{ 1,\dots,h \right\} : T_s(Y_i^*)\in\widehat B_{s,m}
  \text{ for every }s<r, \; T_r({Y}_i^*) \notin  \widehat{B}_{r,m}  \right\}}{h}.
\end{equation}
In sample terms, the power decomposition~\eqref{eq:power:decomposition} becomes:
\begin{equation}
\widehat{\pi}_{m,h} ({P}_1) = \sum_{r = 0}^L \widehat \Delta_{r;m,h} ({P}_1).
\end{equation}
If ${P}_1 = {P}_0$, we get Type I error decomposition instead.

\section{Experiment: augmenting KS with variance and skewness}

The experiment is designed as a direct test of the augmentation
principle rather than as an exhaustive comparison of normality tests.
The Kolmogorov--Smirnov statistic supplies broad empirical-CDF
sensitivity.  Variance and skewness are added because they encode two
simple facts about a standard normal sample: its scale should be close
to one and its distribution should be symmetric.  We ask how much
power these diagnostics add when they receive only a small part of
the KS error budget.  The study uses one fully specified null, one
sample size, and one primary statistic; it demonstrates the allocation
mechanism but does not establish superiority to the wider class of
normality tests.  Comparisons with other omnibus procedures, additional
sample sizes, and composite nulls with estimated parameters are
separate empirical questions. The complete code of the experiment is presented in Appendices~\ref{appendix_a} and~\ref{appendix_b}.

\subsection{Primary and secondary diagnostics}
For arbitrary sample realisation $x^* = (x_1^*, \dots, x^*_n)$, let $\bar{x}^* = \frac{1}{n} \sum_{i = 1}^n x_i^*$ denote the sample mean, and $x_{(i)}^*$ denote the $i$th order statistic of the sample. 

As test statistics we use
\begin{itemize}
\item 
the Kolmogorov--Smirnov distance between the empirical cumulative
distribution function of the sample and the null cumulative
distribution function \(F_0\)
\begin{equation}
T_{\mathrm{KS}}(x^*) = \max_{1\le i\le n} \max
\left(
\frac{i}{n}-F_0(x_{(i)}^*), \,
F_0(x_{(i)}^*)-\frac{i-1}{n}
\right), 
\end{equation}
\item
the sample variance
\begin{equation}
T_{\mathrm{Var}}(x^*)
=
\frac{1}{n-1}
\sum_{i=1}^{n}
\left(x_i^*-\bar{x}^*\right)^2,
\end{equation}
\item
and the sample skewness
\begin{equation}
T_{\mathrm{Skew}}(x^*)
=
\frac{
    \displaystyle
    \frac{1}{n}
    \sum_{i=1}^{n}
    \left(x_i^*-\bar{x}^*\right)^3
}{
    \displaystyle
    \left[
        \frac{1}{n-1}
        \sum_{i=1}^{n}
        \left(x^*_i-\bar{x}^*\right)^2
    \right]^{3/2}
}.
\end{equation}
\end{itemize}

\noindent
For this experiment, \(T_{\mathrm{KS}}\) is the primary statistic,
while \(T_{\mathrm{Var}}\) and \(T_{\mathrm{Skew}}\) are the secondary
diagnostics.

We investigate the four augmented procedures obtained from one or
both secondary statistics:
\begin{equation}
T_{\mathrm{KS}} \to T_{\mathrm{Var}}; \qquad
T_{\mathrm{KS}} \to T_{\mathrm{Skew}}; \qquad
T_{\mathrm{KS}} \to T_{\mathrm{Var}} \to T_{\mathrm{Skew}}; \qquad
T_{\mathrm{KS}} \to T_{\mathrm{Skew}} \to T_{\mathrm{Var}}. 
\end{equation}  

\subsection{Calibration of the augmented test}

We describe the calibration for
\(T_{\mathrm{KS}}\to T_{\mathrm{Var}}\to T_{\mathrm{Skew}}\);
the other chains are constructed analogously.

We parameterize the augmentation on the unconditional budget scale.
Let \(\rho\in[0,1]\) be the fraction of the total level transferred
from KS to the two secondary statistics, and let
\(\omega\in[0,1]\) be the fraction of the transferred amount assigned
to variance.  The first-rejection budgets are
\begin{equation}
\label{eq:experiment-allocation-budgets}
b_{\mathrm{KS}}=\alpha(1-\rho),\qquad
b_{\mathrm{Var}}=\alpha\rho\omega,\qquad
b_{\mathrm{Skew}}=\alpha\rho(1-\omega).
\end{equation}
For the order
\(\mathrm{KS}\to\mathrm{Var}\to\mathrm{Skew}\), the corresponding
conditional levels are obtained from~\eqref{eq:budget-to-gamma}:
\begin{equation}
\label{eq:experiment-allocation-gammas}
\gamma_0=b_{\mathrm{KS}},\qquad
\gamma_1=\frac{b_{\mathrm{Var}}}{1-b_{\mathrm{KS}}},
\qquad
\gamma_2=
\frac{b_{\mathrm{Skew}}}
{1-b_{\mathrm{KS}}-b_{\mathrm{Var}}}.
\end{equation}
For the reverse order, the secondary labels are interchanged.  This
choice yields population first-rejection probabilities equal to the
three budgets in~\eqref{eq:experiment-allocation-budgets} and total
null rejection probability \(\alpha\).  In simulation, the resulting
size is approximate because the acceptance regions are estimated from
a finite calibration sample.

For the fixed null distribution $P_0$
 we simulate:
\begin{equation*}
X_1^*,\ldots,X_m^*
\sim
Q_0.
\end{equation*}
For notational convenience, in this subsection we write
\[
\widehat B_{\mathrm{KS}}
=
\widehat B_{0,m},
\qquad
\widehat B_{\mathrm{Var}}
=
\widehat B_{1,m},
\qquad
\widehat B_{\mathrm{Skew}}
=
\widehat B_{2,m}.
\]
Having $\gamma_0$ and $X_1^*,\ldots,X_m^*$, we find $\widehat{B}_{\mathrm{KS}}$ by solving 
\begin{equation}
\label{eq:gamma_0:KS}
\widehat{\gamma}_{0,m}
:=
\frac{\# \left\{ i \in \left\{ 1,\dots,m \right\}  : T_{\mathrm{KS}} (X_i^*) \notin \widehat{B}_{\mathrm{KS}}\right\}}{ m } \approx \gamma_0.
\end{equation} 
We use a one-sided upper acceptance interval based on the empirical
$(1-\gamma_0)$-quantile of the simulated KS statistics:
\begin{equation}
\widehat{B}_{\mathrm{KS}} = \left(-\infty,\, \widehat{b}_{\mathrm{KS}} \right],
\end{equation}
\begin{equation}
\widehat{b}_{\mathrm{KS}} =
 \mathrm{q}_{1 - \gamma_0} 
 \left( 
 	\bigl(  
		T_{\mathrm{KS}} (X_i^*) 
	 \bigl)_{ i = 1 }^m
\right),
\end{equation}
where $\mathrm{q}_\gamma(A)$ is a quantile of level $\gamma$ for multiset  $A$. 
The implementation uses R's type-7 empirical quantile.  Acceptance
interval endpoints are included, so rejection uses strict violations
of the displayed bounds.  Under the continuous null used here,
boundary ties have probability zero.

Exact equality in~\eqref{eq:gamma_0:KS} need not be attainable because the empirical distribution is discrete and because the empirical quantile convention may interpolate between adjacent order statistics.

Then we determine
\begin{equation}
I_{0,m}
=
\left\{
i \in \left\{ 1,\dots,m \right\} 
:
T_{\mathrm{KS}}(X_i^*)
\in
\widehat{B}_{\mathrm{KS}}
\right\},
\end{equation}
and proceed to solving 
\begin{equation}
\widehat{\gamma}_{1,m}
:=
\frac{\# \left\{ i \in I_{0,m} : T_{\mathrm{Var}} (X_i^*) \notin \widehat{B}_{\mathrm{Var}} \right\}}{ |I_{0,m}| } \approx \gamma_1
\end{equation} 
with equal-tailed two-sided quantile-based acceptance interval
\begin{equation}
\widehat{B}_{\mathrm{Var}} = \left[ \widehat{b}^-_{\mathrm{Var}},\, \widehat{b}^+_{\mathrm{Var}} \right],
\end{equation}
\begin{equation}
\widehat{b}^-_{\mathrm{Var}} = \mathrm{q}_{\gamma_1 / 2} (A_{\mathrm{Var}} ),\qquad 
\widehat{b}^+_{\mathrm{Var}} = \mathrm{q}_{1 - \gamma_1 / 2}(A_{\mathrm{Var}} ),\qquad
A_{\mathrm{Var}} =  \big ( T_{\mathrm{Var}} (X_i^*)  \big )_{i \in I_{0,m}}.
\end{equation}

Finally, we find
\begin{equation}
I_{1,m}
=
\left\{
i \in I_{0,m}
:
T_{\mathrm{Var}}(X_i^*)
\in
\widehat{B}_{\mathrm{Var}}
\right\},
\end{equation}
solve 
\begin{equation}
\widehat{\gamma}_{2,m}
:=
 \frac{\# \left\{ i \in I_{1,m} : T_{\mathrm{Skew}} (X_i^*) \notin \widehat{B}_{\mathrm{Skew}}\right\}}{ |I_{1,m}| } \approx \gamma_2.
\end{equation} 
with
\begin{equation}
\widehat{B}_{\mathrm{Skew}} = \left[ \widehat{b}^-_{\mathrm{Skew}},\, \widehat{b}^+_{\mathrm{Skew}} \right],
\end{equation}
\begin{equation}
\widehat{b}^-_{\mathrm{Skew}} = \mathrm{q}_{\gamma_2 / 2} (A_{\mathrm{Skew}} ),\qquad 
\widehat{b}^+_{\mathrm{Skew}} = \mathrm{q}_{1 - \gamma_2 / 2} (A_{\mathrm{Skew}}),\qquad
A_{\mathrm{Skew}} =  \big( T_{\mathrm{Skew}} (X_i^*)  \big)_{i \in I_{1,m}}.
\end{equation}
The $\widehat{B}_{\mathrm{KS}}$, $\widehat{B}_{\mathrm{Var}}$, $\widehat{B}_{\mathrm{Skew}}$  calibration stage is then complete.

Having obtained all the estimated acceptance regions, we independently
simulate evaluation samples from the alternative distribution:
\[
Y_1^*,\ldots,Y_h^*
\stackrel{\mathrm{iid}}{\sim}
Q_1.
\]
We then estimate the power by
\begin{equation}
\widehat{\pi}_{m,h}(P_1)
=
1-
\frac{
    \#\left\{
        i\in\{1,\ldots,h\}:
        T_{\mathrm{KS}}(Y_i^*)\in\widehat B_{\mathrm{KS}},
        \;
        T_{\mathrm{Var}}(Y_i^*)\in\widehat B_{\mathrm{Var}},
        \;
        T_{\mathrm{Skew}}(Y_i^*)\in\widehat B_{\mathrm{Skew}}
    \right\}
}{
    h
}.
\end{equation}

\noindent
We also compute
\begin{equation}
\widehat{\Delta}_{\mathrm{KS}} ({P}_1) = \frac{\# \left\{ i \in \left\{ 1,\dots,h \right\} : T_{\mathrm{KS}}({Y}_i^*) \notin  \widehat{B}_{\mathrm{KS}}  \right\}}{h},
\end{equation}
\begin{equation}
\widehat{\Delta}_{\mathrm{Var}} ({P}_1) = \frac{\# \left\{ i \in \left\{ 1,\dots,h \right\} : T_{\mathrm{KS}}({Y}_i^*) \in  \widehat{B}_{\mathrm{KS}}, \; T_{\mathrm{Var}}({Y}_i^*) \notin  \widehat{B}_{\mathrm{Var}}  \right\}}{h},
\end{equation}
\begin{equation}
\widehat{\Delta}_{\mathrm{Skew}} ({P}_1) = \frac{\# \left\{ i \in \left\{ 1,\dots,h \right\} : T_{\mathrm{KS}}({Y}_i^*) \in  \widehat{B}_{\mathrm{KS}}, \; T_{\mathrm{Var}}({Y}_i^*) \in  \widehat{B}_{\mathrm{Var}}, \; T_{\mathrm{Skew}}({Y}_i^*) \notin  \widehat{B}_{\mathrm{Skew}} \right\}}{h},
\end{equation}
to get the power decomposition
\begin{equation}
\widehat{\pi}_{m,h} ({P}_1) = \widehat\Delta_{\mathrm{KS}} ({P}_1) + \widehat\Delta_{\mathrm{Var}} ({P}_1) + \widehat\Delta_{\mathrm{Skew}} ({P}_1).
\end{equation}
To get Type I error decomposition we set $P_1 = P_0$:
\begin{equation}
\widehat{\alpha}_{m,h}  = \widehat\Delta_{\mathrm{KS}} ({P}_0) + \widehat\Delta_{\mathrm{Var}} ({P}_0) + \widehat\Delta_{\mathrm{Skew}} ({P}_0).
\end{equation}

\subsection{Power gained from a small secondary allocation}

In all experiments, the null distribution is
\[
P_0=\mathcal N(0,1),
\]
the sample size is \(n=10\), and the nominal level is
\(\alpha=0.05\).  Each of the \(t=100\) repetitions uses \(m=10^6\)
null sample vectors for calibration and \(h=10^6\) independent sample
vectors for evaluation under the null and under each alternative.

The results in Tables~\ref{tab:power-normal}--\ref{tab:null-decomposition}
are taken from the allocation
\[
\rho=0.0075,\qquad \omega=0.5.
\]
We use this point as an interpretable small-budget illustration, not
as an estimated optimal allocation.  In particular, the unweighted
alternative grid below is descriptive and is not a probability model
for future alternatives.  Section~\ref{sec:allocation-sensitivity}
reports the full prespecified grid so that conclusions do not depend
on this single display point.
In terms of unconditional null first-rejection probabilities, the
three-stage chains therefore use
\[
b_0=0.049625,\qquad
b_{\mathrm{Var}}=b_{\mathrm{Skew}}=0.0001875.
\]
Thus \(99.25\%\) of the total \(0.05\) rejection budget remains with
KS and only \(0.75\%\), equal to an absolute probability of
\(0.000375\), is transferred to the two secondary statistics.
For \(\mathrm{KS}\to\mathrm{Var}\to\mathrm{Skew}\), the corresponding
conditional levels are
\[
\gamma_0=0.049625,\qquad
\gamma_1=0.0001972905,\qquad
\gamma_2=0.0001973295.
\]
Because the secondary intervals are equal-tailed, each secondary tail
contains about \(m b_{\mathrm{Var}}/2=93.75\) expected first-rejection
calibration observations.  This effective tail count, rather than
\(m\) alone, explains why finite-\(m\) boundary uncertainty remains a
meaningful diagnostic even with \(m=10^6\).
The reverse order uses the same two secondary conditional levels, with
their statistic labels interchanged.  A two-stage chain uses the same
primary budget \(b_0\) and assigns the complete secondary budget
\(\alpha\rho=0.000375\) to its single secondary statistic; its
secondary conditional level is \(0.0003945811\).  Thus the two-stage
and three-stage procedures use the same total secondary budget, while
the latter divides it equally between variance and skewness.

The jointly calibrated weighted minimum-\(p\) baseline uses weights
\[
(w_{\mathrm{KS}},w_{\mathrm{Var}},w_{\mathrm{Skew}})
=(0.9925,0.00375,0.00375).
\]
As shown in Proposition~\ref{prop:minp-hyperrectangle}, this is also
the corresponding hyperrectangle test in component \(p\)-value space.
For finite Monte Carlo samples, calibration-bank component
\(p\)-values are self-ranks \(r/m\), equivalently leave-one-out
add-one ranks for continuous statistics, while evaluation
\(p\)-values use the add-one denominator \(m+1\).  The minimum-\(p\)
rule uses a strict rejection inequality throughout.  This convention
is asymptotically equivalent to population marginal \(p\)-values; the
independent null bank in Table~\ref{tab:null-decomposition} checks its
finite-\(m\) size directly.
The raw KS, variance, and skewness tests each use their full nominal
level \(0.05\).

In all six tables, a value in parentheses is the Monte Carlo standard
error of the displayed mean across the 100 repetitions.  It equals
the repetition-to-repetition standard deviation divided by
\(\sqrt{100}\).  The use of \(h=10^6\) makes these standard errors
substantially smaller than in the preliminary sensitivity run.
For the variation of one complete calibration/evaluation run, the
relevant standard deviation is therefore ten times the number in
parentheses.  For example, for the three-stage chain it is
\(0.00268\) against \(\mathcal N(0,0.2^2)\),
\(0.00473\) against \(\mathcal N(0,2^2)\), and \(0.00368\)
against \(\operatorname{Cauchy}(0,0.1)\).  These values combine
calibration and evaluation variation; they are not boundary-only
uncertainty measures.

In the tables below,
\[
G_a^\star=\frac{Y-a}{\sqrt a},
\qquad
Y\sim\operatorname{Gamma}(\mathrm{shape}=a,\mathrm{rate}=1).
\]

\subsubsection{Normal location and scale alternatives}

\begin{table}[H]
\centering
\small
\setlength{\tabcolsep}{3.2pt}
\begin{threeparttable}
\caption{Estimated power against normal location and scale alternatives at
\(\rho=0.0075\) and \(\omega=0.5\).}
\label{tab:power-normal}
\begin{tabular}{lcccccccc}
\toprule
Procedure & $\mathcal N(0.1,1)$ & $\mathcal N(0.5,1)$ & $\mathcal N(1,1)$ & $\mathcal N(0,0.1^2)$ & $\mathcal N(0,0.2^2)$ & $\mathcal N(0,1.5^2)$ & $\mathcal N(0,2^2)$ & $\mathcal N(0,3^2)$ \\
\midrule
KS & \bestms{0.0585}{0.00004} & \bestms{0.2742}{0.00008} & \bestms{0.7746}{0.00008} & \ms{0.9903}{0.00003} & \ms{0.4693}{0.00030} & \ms{0.1291}{0.00006} & \ms{0.2463}{0.00007} & \ms{0.4644}{0.00008} \\
Variance & \ms{0.0500}{0.00003} & \ms{0.0500}{0.00003} & \ms{0.0500}{0.00003} & \bestms{1.0000}{0.00000} & \bestms{1.0000}{0.00000} & \bestms{0.4904}{0.00008} & \bestms{0.8552}{0.00005} & \bestms{0.9895}{0.00001} \\
Skewness & \ms{0.0500}{0.00003} & \ms{0.0500}{0.00003} & \ms{0.0500}{0.00003} & \ms{0.0500}{0.00003} & \ms{0.0500}{0.00003} & \ms{0.0500}{0.00003} & \ms{0.0500}{0.00003} & \ms{0.0500}{0.00003} \\
\addlinespace
Weighted min-$p$ & \ms{0.0585}{0.00004} & \ms{0.2737}{0.00008} & \ms{0.7741}{0.00008} & \bestms{1.0000}{0.00000} & \ms{0.9416}{0.00072} & \ms{0.2034}{0.00027} & \ms{0.6009}{0.00048} & \ms{0.9548}{0.00011} \\
\addlinespace
$\mathrm{KS}\rightarrow\mathrm{Var}$ & \ms{0.0584}{0.00004} & \ms{0.2733}{0.00008} & \ms{0.7737}{0.00008} & \bestms{1.0000}{0.00000} & \ms{0.9940}{0.00008} & \ms{0.2252}{0.00024} & \ms{0.6373}{0.00035} & \ms{0.9622}{0.00007} \\
$\mathrm{KS}\rightarrow\mathrm{Skew}$ & \ms{0.0584}{0.00004} & \ms{0.2733}{0.00008} & \ms{0.7736}{0.00008} & \ms{0.9898}{0.00003} & \ms{0.4644}{0.00030} & \ms{0.1288}{0.00006} & \ms{0.2457}{0.00007} & \ms{0.4636}{0.00008} \\
\addlinespace
$\mathrm{KS}\rightarrow\mathrm{Var}\rightarrow\mathrm{Skew}$ & \ms{0.0584}{0.00004} & \ms{0.2733}{0.00008} & \ms{0.7737}{0.00008} & \bestms{1.0000}{0.00000} & \ms{0.9827}{0.00027} & \ms{0.2055}{0.00027} & \ms{0.6049}{0.00047} & \ms{0.9557}{0.00010} \\
$\mathrm{KS}\rightarrow\mathrm{Skew}\rightarrow\mathrm{Var}$ & \ms{0.0584}{0.00004} & \ms{0.2733}{0.00008} & \ms{0.7737}{0.00008} & \bestms{1.0000}{0.00000} & \ms{0.9827}{0.00027} & \ms{0.2055}{0.00027} & \ms{0.6049}{0.00047} & \ms{0.9557}{0.00010} \\
\bottomrule
\end{tabular}
\begin{tablenotes}[flushleft]\footnotesize
\item Entries are mean estimated rejection probabilities with Monte
Carlo standard errors in parentheses.  Bold indicates the largest
unrounded estimated power in a column; exact maxima are all shown in
bold.
\end{tablenotes}
\end{threeparttable}
\end{table}

The raw KS test is strongest among the omnibus procedures for the
three location shifts, while the raw variance and skewness tests
remain near \(0.05\).  Reserving \(0.75\%\) of the Type~I error budget
for secondary statistics has a small cost.  For example, against
\(\mathcal N(0.5,1)\), power changes from \(0.2742\) for KS to
\(0.2733\) for either three-stage ordering.  The corresponding
weighted minimum-\(p\) power is \(0.2737\).

The variance stage produces large gains against scale alternatives.
Against \(\mathcal N(0,0.2^2)\), power is \(0.4693\) for KS,
\(0.9416\) for weighted minimum-\(p\), and \(0.9827\) for either
three-stage chain.  Against \(\mathcal N(0,2^2)\), the corresponding
values are \(0.2463\), \(0.6009\), and \(0.6049\).  The skewness-only
chain is essentially no better than KS for these symmetric
alternatives.

The two-stage variance chain is stronger than the three-stage chain
for several scale alternatives: for example, its power against
\(\mathcal N(0,2^2)\) is \(0.6373\).  This is expected because the
two-stage procedure assigns the entire secondary budget to variance,
whereas the three-stage procedure assigns half to variance and half
to skewness.  The raw variance test remains the strongest specialized
test because it assigns its entire level \(0.05\) to scale detection.
The comparison therefore illustrates the intended trade-off:
substantial additional scale sensitivity is obtained while retaining
almost all of the primary KS power.

\subsubsection{Cauchy and standardized gamma alternatives}

\begin{table}[H]
\centering
\small
\setlength{\tabcolsep}{3.2pt}
\begin{threeparttable}
\caption{Estimated power against Cauchy and standardized gamma alternatives
at \(\rho=0.0075\) and \(\omega=0.5\).}
\label{tab:power-cauchy-gamma}
\begin{tabular}{lccccccc}
\toprule
Procedure & $\operatorname{Cauchy}(0,0.1)$ & $\operatorname{Cauchy}(0,0.5)$ & $\operatorname{Cauchy}(0,1)$ & $G_{0.1}^{\star}$ & $G_{0.5}^{\star}$ & $G_{1}^{\star}$ & $G_{10}^{\star}$ \\
\midrule
KS & \ms{0.4454}{0.00016} & \ms{0.0445}{0.00003} & \ms{0.1269}{0.00006} & \ms{0.5275}{0.00012} & \ms{0.1837}{0.00007} & \ms{0.1167}{0.00005} & \ms{0.0563}{0.00003} \\
Variance & \bestms{0.7850}{0.00005} & \bestms{0.6473}{0.00005} & \bestms{0.8910}{0.00003} & \ms{0.6825}{0.00006} & \ms{0.3609}{0.00008} & \ms{0.2338}{0.00007} & \ms{0.0708}{0.00004} \\
Skewness & \ms{0.5739}{0.00007} & \ms{0.5739}{0.00006} & \ms{0.5740}{0.00006} & \bestms{0.9206}{0.00004} & \bestms{0.5555}{0.00011} & \bestms{0.3617}{0.00010} & \bestms{0.0865}{0.00004} \\
\addlinespace
Weighted min-$p$ & \ms{0.7027}{0.00038} & \ms{0.5169}{0.00015} & \ms{0.8029}{0.00015} & \ms{0.7169}{0.00028} & \ms{0.2486}{0.00018} & \ms{0.1452}{0.00010} & \ms{0.0576}{0.00003} \\
\addlinespace
$\mathrm{KS}\rightarrow\mathrm{Var}$ & \ms{0.6810}{0.00025} & \ms{0.5159}{0.00012} & \ms{0.8129}{0.00011} & \ms{0.6320}{0.00016} & \ms{0.2194}{0.00009} & \ms{0.1371}{0.00007} & \ms{0.0577}{0.00003} \\
$\mathrm{KS}\rightarrow\mathrm{Skew}$ & \ms{0.6015}{0.00015} & \ms{0.3138}{0.00018} & \ms{0.3732}{0.00017} & \ms{0.7038}{0.00023} & \ms{0.2577}{0.00020} & \ms{0.1471}{0.00012} & \ms{0.0574}{0.00003} \\
\addlinespace
$\mathrm{KS}\rightarrow\mathrm{Var}\rightarrow\mathrm{Skew}$ & \ms{0.7407}{0.00037} & \ms{0.5232}{0.00014} & \ms{0.8045}{0.00015} & \ms{0.7441}{0.00029} & \ms{0.2595}{0.00021} & \ms{0.1504}{0.00012} & \ms{0.0579}{0.00003} \\
$\mathrm{KS}\rightarrow\mathrm{Skew}\rightarrow\mathrm{Var}$ & \ms{0.7407}{0.00037} & \ms{0.5232}{0.00014} & \ms{0.8045}{0.00015} & \ms{0.7441}{0.00029} & \ms{0.2595}{0.00021} & \ms{0.1504}{0.00012} & \ms{0.0579}{0.00003} \\
\bottomrule
\end{tabular}
\begin{tablenotes}[flushleft]\footnotesize
\item Entries are mean estimated rejection probabilities with Monte
Carlo standard errors in parentheses.  Bold indicates the largest
estimated power in a column.
\end{tablenotes}
\end{threeparttable}
\end{table}

For the Cauchy alternatives, the raw variance test is the strongest
specialized procedure, with power between \(0.6473\) and \(0.8910\).
The three-stage chain nevertheless improves sharply on KS: against
\(\operatorname{Cauchy}(0,0.1)\), its power is \(0.7407\), compared
with \(0.4454\) for KS and \(0.7027\) for weighted minimum-\(p\).
Against \(\operatorname{Cauchy}(0,0.5)\), the corresponding values are
\(0.5232\), \(0.0445\), and \(0.5169\).  Both variance and finite-sample
skewness can identify heavy-tailed samples that survive the KS stage.

The relative performance of the two- and three-stage procedures
depends on whether the second secondary statistic compensates for
splitting the budget.  Against \(\operatorname{Cauchy}(0,0.1)\), the
three-stage power \(0.7407\) exceeds the variance-only chain power
\(0.6810\).  Against \(\operatorname{Cauchy}(0,1)\), however, the
variance-only chain has power \(0.8129\), slightly above the
three-stage value \(0.8045\).

For the standardized gamma alternatives, the raw skewness test is
strongest, as expected for asymmetric distributions.  The three-stage
chain raises power against \(G_{0.1}^{\star}\) from \(0.5275\) for KS
to \(0.7441\); weighted minimum-\(p\) reaches \(0.7169\).  Against
\(G_{0.5}^{\star}\), the corresponding values are \(0.1837\),
\(0.2595\), and \(0.2486\).  All procedures approach the nominal level
as the gamma shape increases and the standardized gamma distribution
approaches normality.

\subsubsection{Student alternatives}

\begin{table}[H]
\centering
\small
\setlength{\tabcolsep}{3.2pt}
\begin{threeparttable}
\caption{Estimated power against Student alternatives at
\(\rho=0.0075\) and \(\omega=0.5\).}
\label{tab:power-student}
\begin{tabular}{lcccc}
\toprule
Procedure & $t_1$ & $t_2$ & $t_3$ & $t_4$ \\
\midrule
KS & \ms{0.1269}{0.00005} & \ms{0.0791}{0.00004} & \ms{0.0670}{0.00004} & \ms{0.0618}{0.00004} \\
Variance & \bestms{0.8910}{0.00003} & \bestms{0.6120}{0.00006} & \bestms{0.4232}{0.00007} & \bestms{0.3123}{0.00007} \\
Skewness & \ms{0.5740}{0.00006} & \ms{0.3190}{0.00006} & \ms{0.2126}{0.00006} & \ms{0.1608}{0.00005} \\
\addlinespace
Weighted min-$p$ & \ms{0.8028}{0.00015} & \ms{0.4197}{0.00025} & \ms{0.2322}{0.00021} & \ms{0.1492}{0.00016} \\
\addlinespace
$\mathrm{KS}\rightarrow\mathrm{Var}$ & \ms{0.8129}{0.00011} & \ms{0.4377}{0.00019} & \ms{0.2473}{0.00017} & \ms{0.1602}{0.00014} \\
$\mathrm{KS}\rightarrow\mathrm{Skew}$ & \ms{0.3732}{0.00017} & \ms{0.1561}{0.00011} & \ms{0.0967}{0.00006} & \ms{0.0758}{0.00005} \\
\addlinespace
$\mathrm{KS}\rightarrow\mathrm{Var}\rightarrow\mathrm{Skew}$ & \ms{0.8044}{0.00015} & \ms{0.4227}{0.00025} & \ms{0.2349}{0.00021} & \ms{0.1512}{0.00016} \\
$\mathrm{KS}\rightarrow\mathrm{Skew}\rightarrow\mathrm{Var}$ & \ms{0.8044}{0.00015} & \ms{0.4227}{0.00025} & \ms{0.2349}{0.00021} & \ms{0.1512}{0.00016} \\
\bottomrule
\end{tabular}
\begin{tablenotes}[flushleft]\footnotesize
\item Entries are mean estimated rejection probabilities with Monte
Carlo standard errors in parentheses.  Bold indicates the largest
estimated power in a column.
\end{tablenotes}
\end{threeparttable}
\end{table}

All four Student alternatives are symmetric but heavy-tailed.  The raw
variance test is strongest throughout, and the variance stage supplies
most of the chain's gain over KS.  Against \(t_1\), power rises from
\(0.1269\) for KS to \(0.8044\) for the three-stage chain and \(0.8028\)
for weighted minimum-\(p\).  Against \(t_4\), it rises from \(0.0618\)
to \(0.1512\) and \(0.1492\), respectively.

The variance-only chain is somewhat stronger than the three-stage
chain because it receives the full secondary budget: its powers are
\(0.8129\), \(0.4377\), \(0.2473\), and \(0.1602\).  Skewness still
has finite-sample power against these symmetric alternatives because
heavy tails make extreme values of the two-sided sample skewness
statistic more frequent, but allocating half of the secondary budget
to skewness is less efficient for this family.

\subsubsection{Stagewise decomposition for
\texorpdfstring{\(\mathrm{KS}\rightarrow\mathrm{Var}
\rightarrow\mathrm{Skew}\)}{KS to Var to Skew}}

\begin{table}[p]
\centering
\small
\setlength{\tabcolsep}{2.5pt}
\begin{threeparttable}
\caption{Stagewise power decomposition for
\(\mathrm{KS}\rightarrow\mathrm{Var}\rightarrow\mathrm{Skew}\) at
\(\rho=0.0075\) and \(\omega=0.5\).}
\label{tab:decomp-var-skew}
\begin{tabular}{lccccc}
\toprule
Alternative & Total & $\Delta_{\mathrm{KS}}$ & $\Delta_{\mathrm{Var}}$ & $\Delta_{\mathrm{Skew}}$ & Min-$p$ total \\
\midrule
\multicolumn{6}{l}{\textit{Normal alternatives}} \\
$\mathcal N(0.1,1)$ & \ms{0.058441}{0.000037} & \ms{0.058069}{0.000037} & \ms{0.000186}{0.000002} & \ms{0.000185}{0.000002} & \ms{0.058453}{0.000037} \\
$\mathcal N(0.5,1)$ & \ms{0.273314}{0.000078} & \ms{0.273106}{0.000078} & \ms{0.000112}{0.000001} & \ms{0.000096}{0.000001} & \ms{0.273691}{0.000078} \\
$\mathcal N(1,1)$ & \ms{0.773664}{0.000077} & \ms{0.773608}{0.000077} & \ms{0.000043}{0.000001} & \ms{0.000013}{0.000000} & \ms{0.774069}{0.000078} \\
$\mathcal N(0,0.1^2)$ & \ms{1.000000}{0.000000} & \ms{0.989752}{0.000035} & \ms{0.010248}{0.000035} & \ms{0.000000}{0.000000} & \ms{1.000000}{0.000000} \\
$\mathcal N(0,0.2^2)$ & \ms{0.982654}{0.000268} & \ms{0.464326}{0.000303} & \ms{0.518325}{0.000392} & \ms{0.000002}{0.000000} & \ms{0.941564}{0.000717} \\
$\mathcal N(0,1.5^2)$ & \ms{0.205459}{0.000275} & \ms{0.128405}{0.000060} & \ms{0.076889}{0.000265} & \ms{0.000165}{0.000002} & \ms{0.203351}{0.000270} \\
$\mathcal N(0,2^2)$ & \ms{0.604861}{0.000473} & \ms{0.245326}{0.000073} & \ms{0.359453}{0.000474} & \ms{0.000082}{0.000001} & \ms{0.600877}{0.000482} \\
$\mathcal N(0,3^2)$ & \ms{0.955661}{0.000104} & \ms{0.463359}{0.000082} & \ms{0.492293}{0.000128} & \ms{0.000009}{0.000000} & \ms{0.954794}{0.000107} \\
\addlinespace
\multicolumn{6}{l}{\textit{Cauchy and standardized gamma alternatives}} \\
$\operatorname{Cauchy}(0,0.1)$ & \ms{0.740705}{0.000368} & \ms{0.442982}{0.000155} & \ms{0.216193}{0.000334} & \ms{0.081531}{0.000125} & \ms{0.702721}{0.000377} \\
$\operatorname{Cauchy}(0,0.5)$ & \ms{0.523247}{0.000142} & \ms{0.044133}{0.000029} & \ms{0.460651}{0.000156} & \ms{0.018463}{0.000068} & \ms{0.516879}{0.000150} \\
$\operatorname{Cauchy}(0,1)$ & \ms{0.804480}{0.000146} & \ms{0.126300}{0.000056} & \ms{0.676935}{0.000155} & \ms{0.001246}{0.000007} & \ms{0.802869}{0.000148} \\
$G_{0.1}^{\star}$ & \ms{0.744141}{0.000291} & \ms{0.525592}{0.000126} & \ms{0.092734}{0.000177} & \ms{0.125815}{0.000267} & \ms{0.716943}{0.000283} \\
$G_{0.5}^{\star}$ & \ms{0.259501}{0.000208} & \ms{0.182784}{0.000073} & \ms{0.031250}{0.000061} & \ms{0.045467}{0.000209} & \ms{0.248650}{0.000179} \\
$G_{1}^{\star}$ & \ms{0.150411}{0.000117} & \ms{0.116042}{0.000051} & \ms{0.017621}{0.000046} & \ms{0.016748}{0.000104} & \ms{0.145179}{0.000097} \\
$G_{10}^{\star}$ & \ms{0.057904}{0.000033} & \ms{0.055882}{0.000031} & \ms{0.001213}{0.000007} & \ms{0.000808}{0.000008} & \ms{0.057611}{0.000032} \\
\addlinespace
\multicolumn{6}{l}{\textit{Student alternatives}} \\
$t_1$ & \ms{0.804415}{0.000152} & \ms{0.126303}{0.000046} & \ms{0.676872}{0.000159} & \ms{0.001241}{0.000008} & \ms{0.802815}{0.000152} \\
$t_2$ & \ms{0.422706}{0.000249} & \ms{0.078607}{0.000037} & \ms{0.341488}{0.000253} & \ms{0.002610}{0.000016} & \ms{0.419714}{0.000249} \\
$t_3$ & \ms{0.234856}{0.000212} & \ms{0.066523}{0.000037} & \ms{0.165776}{0.000214} & \ms{0.002556}{0.000016} & \ms{0.232244}{0.000211} \\
$t_4$ & \ms{0.151159}{0.000162} & \ms{0.061376}{0.000039} & \ms{0.087625}{0.000160} & \ms{0.002159}{0.000014} & \ms{0.149163}{0.000162} \\
\bottomrule
\end{tabular}
\begin{tablenotes}[flushleft]\footnotesize
\item Entries are means with Monte Carlo standard errors in
parentheses.  The three first-rejection contributions sum to the chain
total within every repetition.  The last column is a competing total,
not part of the decomposition.
\end{tablenotes}
\end{threeparttable}
\end{table}

Table~\ref{tab:decomp-var-skew} shows where the chain's rejections
occur.  For location shifts, nearly all power comes from KS.  Against
\(\mathcal N(1,1)\), the decomposition is
\[
0.773608+0.000043+0.000013=0.773664.
\]
For normal scale alternatives, variance supplies essentially all of
the secondary gain.  For example,
\[
0.464326+0.518325+0.000002=0.982654
\]
against \(\mathcal N(0,0.2^2)\), up to rounding.

Both secondary statistics can matter for Cauchy alternatives.  Against
\(\operatorname{Cauchy}(0,0.1)\), variance contributes \(0.216193\)
and skewness contributes \(0.081531\) after the KS contribution
\(0.442982\).  For the gamma alternatives, skewness is more prominent:
against \(G_{0.1}^{\star}\), its contribution \(0.125815\) exceeds the
variance contribution \(0.092734\).  For the Student alternatives,
variance accounts for almost all secondary-stage power.

\subsubsection{Stagewise decomposition for
\texorpdfstring{\(\mathrm{KS}\rightarrow\mathrm{Skew}
\rightarrow\mathrm{Var}\)}{KS to Skew to Var}}

\begin{table}[p]
\centering
\small
\setlength{\tabcolsep}{2.5pt}
\begin{threeparttable}
\caption{Stagewise power decomposition for
\(\mathrm{KS}\rightarrow\mathrm{Skew}\rightarrow\mathrm{Var}\) at
\(\rho=0.0075\) and \(\omega=0.5\).}
\label{tab:decomp-skew-var}
\begin{tabular}{lccccc}
\toprule
Alternative & Total & $\Delta_{\mathrm{KS}}$ & $\Delta_{\mathrm{Skew}}$ & $\Delta_{\mathrm{Var}}$ & Min-$p$ total \\
\midrule
\multicolumn{6}{l}{\textit{Normal alternatives}} \\
$\mathcal N(0.1,1)$ & \ms{0.058441}{0.000037} & \ms{0.058069}{0.000037} & \ms{0.000185}{0.000002} & \ms{0.000186}{0.000002} & \ms{0.058453}{0.000037} \\
$\mathcal N(0.5,1)$ & \ms{0.273314}{0.000078} & \ms{0.273106}{0.000078} & \ms{0.000096}{0.000001} & \ms{0.000112}{0.000001} & \ms{0.273691}{0.000078} \\
$\mathcal N(1,1)$ & \ms{0.773664}{0.000077} & \ms{0.773608}{0.000077} & \ms{0.000013}{0.000000} & \ms{0.000043}{0.000001} & \ms{0.774069}{0.000078} \\
$\mathcal N(0,0.1^2)$ & \ms{1.000000}{0.000000} & \ms{0.989752}{0.000035} & \ms{0.000000}{0.000000} & \ms{0.010248}{0.000035} & \ms{1.000000}{0.000000} \\
$\mathcal N(0,0.2^2)$ & \ms{0.982654}{0.000268} & \ms{0.464326}{0.000303} & \ms{0.000015}{0.000000} & \ms{0.518313}{0.000392} & \ms{0.941564}{0.000717} \\
$\mathcal N(0,1.5^2)$ & \ms{0.205463}{0.000275} & \ms{0.128405}{0.000060} & \ms{0.000183}{0.000002} & \ms{0.076875}{0.000266} & \ms{0.203351}{0.000270} \\
$\mathcal N(0,2^2)$ & \ms{0.604868}{0.000474} & \ms{0.245326}{0.000073} & \ms{0.000168}{0.000002} & \ms{0.359374}{0.000474} & \ms{0.600877}{0.000482} \\
$\mathcal N(0,3^2)$ & \ms{0.955663}{0.000104} & \ms{0.463359}{0.000082} & \ms{0.000140}{0.000001} & \ms{0.492163}{0.000128} & \ms{0.954794}{0.000107} \\
\addlinespace
\multicolumn{6}{l}{\textit{Cauchy and standardized gamma alternatives}} \\
$\operatorname{Cauchy}(0,0.1)$ & \ms{0.740705}{0.000368} & \ms{0.442982}{0.000155} & \ms{0.145439}{0.000139} & \ms{0.152284}{0.000323} & \ms{0.702721}{0.000377} \\
$\operatorname{Cauchy}(0,0.5)$ & \ms{0.523249}{0.000142} & \ms{0.044133}{0.000029} & \ms{0.247322}{0.000225} & \ms{0.231794}{0.000225} & \ms{0.516879}{0.000150} \\
$\operatorname{Cauchy}(0,1)$ & \ms{0.804482}{0.000146} & \ms{0.126300}{0.000056} & \ms{0.226389}{0.000210} & \ms{0.451793}{0.000265} & \ms{0.802869}{0.000148} \\
$G_{0.1}^{\star}$ & \ms{0.744142}{0.000291} & \ms{0.525592}{0.000126} & \ms{0.158757}{0.000297} & \ms{0.059792}{0.000167} & \ms{0.716943}{0.000283} \\
$G_{0.5}^{\star}$ & \ms{0.259502}{0.000208} & \ms{0.182784}{0.000073} & \ms{0.058326}{0.000231} & \ms{0.018392}{0.000051} & \ms{0.248650}{0.000179} \\
$G_{1}^{\star}$ & \ms{0.150411}{0.000117} & \ms{0.116042}{0.000051} & \ms{0.022016}{0.000118} & \ms{0.012353}{0.000040} & \ms{0.145179}{0.000097} \\
$G_{10}^{\star}$ & \ms{0.057904}{0.000033} & \ms{0.055882}{0.000031} & \ms{0.000865}{0.000008} & \ms{0.001157}{0.000007} & \ms{0.057611}{0.000032} \\
\addlinespace
\multicolumn{6}{l}{\textit{Student alternatives}} \\
$t_1$ & \ms{0.804418}{0.000152} & \ms{0.126303}{0.000046} & \ms{0.226328}{0.000211} & \ms{0.451787}{0.000263} & \ms{0.802815}{0.000152} \\
$t_2$ & \ms{0.422710}{0.000249} & \ms{0.078607}{0.000037} & \ms{0.065403}{0.000117} & \ms{0.278700}{0.000279} & \ms{0.419714}{0.000249} \\
$t_3$ & \ms{0.234859}{0.000212} & \ms{0.066523}{0.000037} & \ms{0.023586}{0.000063} & \ms{0.144750}{0.000220} & \ms{0.232244}{0.000211} \\
$t_4$ & \ms{0.151162}{0.000162} & \ms{0.061376}{0.000039} & \ms{0.010547}{0.000035} & \ms{0.079239}{0.000161} & \ms{0.149163}{0.000162} \\
\bottomrule
\end{tabular}
\begin{tablenotes}[flushleft]\footnotesize
\item Entries are means with Monte Carlo standard errors in
parentheses.  The three first-rejection contributions sum to the chain
total within every repetition.  The last column is a competing total,
not part of the decomposition.
\end{tablenotes}
\end{threeparttable}
\end{table}

Reversing the two secondary statistics has almost no effect on total
power at this allocation, but it changes attribution.  Against
\(\operatorname{Cauchy}(0,0.5)\), placing variance first attributes
\(0.460651\) to variance and \(0.018463\) to skewness.  Placing
skewness first attributes \(0.247322\) to skewness and \(0.231794\)
to variance.  Similarly, against \(t_1\), the
variance-first decomposition assigns \(0.676872\) to variance and
\(0.001241\) to skewness, whereas the reverse order assigns
\(0.226328\) to skewness and \(0.451787\) to variance.

This is a first-rejection decomposition: when a sample violates both
secondary bounds, the earlier statistic receives the attribution.
Consequently, the contributions are intrinsically order-dependent
even when the union of the rejection regions, and hence total power,
is nearly unchanged.  The simulations do not establish the
conditional independence assumed by Proposition~\ref{prop:order-invariance};
they show only that its order-invariance conclusion is an excellent
approximation for the present statistics and allocation.

\subsubsection{Null rejection probabilities and Type~I error decomposition}

\begin{table}[H]
\centering
\small
\begin{threeparttable}
\caption{Null rejection probabilities and stagewise Type~I error
decomposition at \(\rho=0.0075\) and \(\omega=0.5\).}
\label{tab:null-decomposition}
\begin{tabular}{lcc}
\toprule
Procedure or contribution & Mean & Monte Carlo SE \\
\midrule
\multicolumn{3}{l}{\textit{Raw tests and competing baseline}} \\
KS total & 0.04999491 & 0.00003038 \\
Variance total & 0.04996463 & 0.00002977 \\
Skewness total & 0.05000379 & 0.00002900 \\
Weighted minimum-$p$ total & 0.04999247 & 0.00003081 \\
\addlinespace
\multicolumn{3}{l}{\textit{Two-stage chain $\mathrm{KS}\rightarrow\mathrm{Var}$}} \\
Total & 0.04999282 & 0.00003047 \\
$\Delta_{\mathrm{KS}}$ & 0.04962018 & 0.00003027 \\
$\Delta_{\mathrm{Var}}$ & 0.00037264 & 0.00000273 \\
\addlinespace
\multicolumn{3}{l}{\textit{Two-stage chain $\mathrm{KS}\rightarrow\mathrm{Skew}$}} \\
Total & 0.04999839 & 0.00002976 \\
$\Delta_{\mathrm{KS}}$ & 0.04962018 & 0.00003027 \\
$\Delta_{\mathrm{Skew}}$ & 0.00037821 & 0.00000287 \\
\addlinespace
\multicolumn{3}{l}{\textit{Three-stage chain $\mathrm{KS}\rightarrow\mathrm{Var}\rightarrow\mathrm{Skew}$}} \\
Total & 0.04999782 & 0.00003025 \\
$\Delta_{\mathrm{KS}}$ & 0.04962018 & 0.00003027 \\
$\Delta_{\mathrm{Var}}$ & 0.00018720 & 0.00000213 \\
$\Delta_{\mathrm{Skew}}$ & 0.00019044 & 0.00000191 \\
\addlinespace
\multicolumn{3}{l}{\textit{Three-stage chain $\mathrm{KS}\rightarrow\mathrm{Skew}\rightarrow\mathrm{Var}$}} \\
Total & 0.04999782 & 0.00003025 \\
$\Delta_{\mathrm{KS}}$ & 0.04962018 & 0.00003027 \\
$\Delta_{\mathrm{Skew}}$ & 0.00019047 & 0.00000191 \\
$\Delta_{\mathrm{Var}}$ & 0.00018717 & 0.00000213 \\
\bottomrule
\end{tabular}
\begin{tablenotes}[flushleft]\footnotesize
\item The displayed uncertainty is the standard error of the mean
over 100 independent calibration/evaluation repetitions.  The
three-stage population targets are \(0.049625\), \(0.0001875\), and
\(0.0001875\).
\end{tablenotes}
\end{threeparttable}
\end{table}

All raw tests, chains, and the weighted minimum-\(p\) baseline have
estimated null rejection probabilities close to \(0.05\).  For the
variance-first three-stage chain, the empirical decomposition is
\[
0.04962018+0.00018720+0.00019044=0.04999782.
\]
The target decomposition is
\[
0.049625+0.0001875+0.0001875=0.05.
\]
The total deviation, \(-2.18\times10^{-6}\), is much smaller than its
Monte Carlo standard error \(3.025\times10^{-5}\).  The reverse order
has the same total to the reported precision and nearly identical
statistic-specific contributions.  The two-stage secondary
contributions are also close to their common target \(0.000375\).
These results provide a precise numerical check that the allocation
and primary-level adjustment preserve the intended overall size.

\subsubsection{Sensitivity to the Type~I error allocation}
\label{sec:allocation-sensitivity}

To examine the power trade-off away from the allocation used in
Tables~\ref{tab:power-normal}--\ref{tab:null-decomposition}, define
the budgets through \((\rho,\omega)\) exactly as in
\eqref{eq:experiment-allocation-budgets}, with conditional levels
obtained from~\eqref{eq:budget-to-gamma}.  Thus \(\rho\) controls the
total amount removed from KS, while \(\omega\) controls only its
division between variance and skewness.  For the reverse order,
variance and skewness are interchanged.

We considered
\[
\rho\in\{0,0.0025,0.005,0.0075,0.01,0.02,0.05,0.10,
0.15,0.20,0.30,0.40,0.50\}
\]
and
\[
\omega\in\{1,0.75,0.50,0.25,0\}.
\]
The full simulation design was used at every allocation:
\(m=h=10^6\) and \(t=100\).  Within a repetition, the same calibration
and evaluation samples were reused for all allocations and
procedures, so comparisons across \((\rho,\omega)\), between chain
orders, and with weighted minimum-\(p\) are paired.

For the baseline, the weights are
\[
(w_{\mathrm{KS}},w_{\mathrm{Var}},w_{\mathrm{Skew}})
=(1-\rho,\rho\omega,\rho(1-\omega)),
\]
and the empirical null distribution jointly calibrates the rejection
threshold of
\[
M_w=\min_{j:w_j>0}\frac{p_j}{w_j}.
\]
This is the weighted minimum-\(p\)/hyperrectangle baseline described
in Proposition~\ref{prop:minp-hyperrectangle}.

Across all 65 allocation pairs, mean null rejection probabilities
ranged from \(0.04997535\) to \(0.05003615\) for
\(\mathrm{KS}\to\mathrm{Var}\to\mathrm{Skew}\), from \(0.04997535\)
to \(0.05003529\) for the reverse order, and from \(0.04997624\) to
\(0.05003495\) for weighted minimum-\(p\).  Their Monte Carlo standard
errors ranged from \(2.83\times10^{-5}\) to \(3.39\times10^{-5}\).
The largest standardized deviation from \(0.05\) was \(1.18\)
standard errors.  The largest absolute discrepancy between an
estimated three-stage contribution and its target budget was
\(3.84\times10^{-5}\).  There is therefore no indication that the
power comparisons are driven by material size distortion.

Table~\ref{tab:allocation-sensitivity-main} summarizes the main
trade-off for the equal split \(\omega=0.5\).  Moving from \(\rho=0\)
to \(\rho=0.0075\) lowers average power over the three location
alternatives by \(0.000640\), from \(0.369113\) to \(0.368473\).
Over the other 16 alternatives, average power rises from \(0.258453\)
to \(0.540135\).  Increasing \(\rho\) further continues to improve the
latter average, but the location-power cost becomes progressively
larger.

\begin{table}[H]
\centering
\small
\begin{threeparttable}
\caption{Power sensitivity to the secondary-budget fraction for
\(\omega=0.5\).}
\label{tab:allocation-sensitivity-main}
\begin{tabular}{rrrrrr}
\toprule
\(\rho\) & \(b_0\) & Location & Other & Chain, all & Min-\(p\), all \\
\midrule
0      & 0.050000 & 0.3691 & 0.2585 & 0.2759 & 0.2759 \\
0.0075 & 0.049625 & 0.3685 & 0.5401 & 0.5130 & 0.5053 \\
0.0200 & 0.049000 & 0.3674 & 0.5617 & 0.5310 & 0.5252 \\
0.0500 & 0.047500 & 0.3648 & 0.5846 & 0.5499 & 0.5451 \\
0.1000 & 0.045000 & 0.3602 & 0.6044 & 0.5659 & 0.5618 \\
0.2000 & 0.040000 & 0.3504 & 0.6267 & 0.5831 & 0.5802 \\
0.5000 & 0.025000 & 0.3127 & 0.6611 & 0.6061 & 0.6057 \\
\bottomrule
\end{tabular}
\begin{tablenotes}[flushleft]\footnotesize
\item ``Location'' is the unweighted mean over the three normal
location alternatives, ``Other'' is the unweighted mean over the
remaining 16 alternatives, and ``all'' is the unweighted mean over all
19 alternatives.  These descriptive averages depend on the selected
grid of alternatives and are not an optimization criterion.
\end{tablenotes}
\end{threeparttable}
\end{table}

At the selected allocation, the chain-minus-minimum-\(p\) difference
averaged over all 19 alternatives is \(0.007684\).  The family-average
differences are \(-0.000265\) for location,
\(0.009610\) for normal scale, \(0.015321\) for Cauchy,
\(0.010894\) for standardized gamma, and \(0.002300\) for Student
alternatives.  This is not a uniform dominance result: the weighted
minimum-\(p\) baseline is slightly stronger for the location family,
whereas the conditional chain is stronger on average for the four
families motivating the secondary statistics.

Because these procedures were evaluated on the same simulated
samples, the uncertainty of their difference is smaller than would be
obtained by treating the two power estimates as independent.
Table~\ref{tab:paired-chain-minp} reports representative paired
comparisons at the main allocation.

\begin{table}[H]
\centering
\small
\begin{threeparttable}
\caption{Selected paired power differences between
\(\mathrm{KS}\rightarrow\mathrm{Var}\rightarrow\mathrm{Skew}\) and
weighted minimum-\(p\) at \(\rho=0.0075\), \(\omega=0.5\).}
\label{tab:paired-chain-minp}
\begin{tabular}{lrr}
\toprule
Alternative & Chain minus min-\(p\) & Paired Monte Carlo SE \\
\midrule
\(\mathcal N(0.5,1)\) & \(-0.00037748\) & \(0.00000507\) \\
\(\mathcal N(0,0.2^2)\) & \(0.04109003\) & \(0.00071371\) \\
\(\mathcal N(0,2^2)\) & \(0.00398381\) & \(0.00017931\) \\
\(\operatorname{Cauchy}(0,0.1)\) & \(0.03798418\) & \(0.00044298\) \\
\(\operatorname{Cauchy}(0,0.5)\) & \(0.00636827\) & \(0.00007722\) \\
\(G_{0.1}^{\star}\) & \(0.02719794\) & \(0.00031073\) \\
\(G_{0.5}^{\star}\) & \(0.01085162\) & \(0.00019006\) \\
\(t_1\) & \(0.00160015\) & \(0.00005380\) \\
\(t_4\) & \(0.00199605\) & \(0.00005718\) \\
\bottomrule
\end{tabular}
\begin{tablenotes}[flushleft]\footnotesize
\item Positive values favor the chain.  The standard errors are
computed from the within-repetition differences and quantify Monte
Carlo uncertainty, not uncertainty over a population of alternative
distributions.
\end{tablenotes}
\end{threeparttable}
\end{table}

The paired results confirm that the baseline's small advantage for the
location example and the chain's advantages for the displayed scale,
tail, and skewness examples are larger than Monte Carlo noise.  They
do not establish uniform dominance: the sign and magnitude depend on
the alternative and on the boundary-allocation rule.

The split parameter also matters.  At \(\rho=0.0075\), the chain's
overall averages for
\(\omega=1,0.75,0.5,0.25,0\) are \(0.507275\), \(0.515799\),
\(0.513031\), \(0.505961\), and \(0.344946\), respectively.  The
value \(0.75\) is slightly best for this particular unweighted
aggregate, while the equal split is within \(0.002768\) and provides
a neutral allocation between the two secondary features.  The family
results in Appendix~\ref{appendix_sensitivity} show the expected
trade-off: variance-heavy allocations favor scale and Student
alternatives, while more skewness budget helps the standardized gamma
family.

Finally, the two secondary-statistic orders remain almost
indistinguishable near the main allocation.  Across all 19 alternatives
and all five values of \(\omega\) at \(\rho=0.0075\), the largest
absolute paired mean difference is \(7.18\times10^{-6}\); its paired
Monte Carlo standard error is \(7.23\times10^{-6}\).  Thus the largest
observed difference near the main allocation is itself comparable to
simulation noise.  Over the full and deliberately aggressive grid,
the largest difference is \(0.00138557\), at
\((\rho,\omega)=(0.5,0.25)\) for
\(\mathcal N(0,1.5^2)\), with paired Monte Carlo standard error
\(0.00002599\).  Approximate order invariance is therefore very
accurate at the small allocation used in the main tables, but it
should not be asserted uniformly over arbitrarily large secondary
budgets.

\subsubsection{Overall conclusions from the experiment}

Within this focused design, the experiment supports the paper's
proof-of-principle claim.  Transferring only \(0.75\%\) of the total
error budget from KS to variance and skewness has a very small effect
on location power but produces large gains against several scale,
heavy-tailed, and asymmetric alternatives.  The chain is useful here
because it measures both sides of that exchange: the reduced KS
contribution and the additional first-rejection probability supplied
by each secondary diagnostic.

The result is not a uniform dominance statement.  A specialized raw
test, or a two-stage chain that assigns the entire secondary budget to
the relevant feature, can remain stronger for a particular
alternative family.  The three-stage chain is instead a broader
omnibus compromise that protects the primary procedure while adding
limited feature-specific safeguards.

Finally, total power and stagewise attribution are different.
Reversing variance and skewness barely changes total power at the
selected allocation, but it can materially change which statistic
receives first-rejection credit.  The decomposition should therefore
be interpreted as a measurement of incremental contribution under a
prespecified order, not as an order-free ranking of the statistics.

\section{Conclusion}

The paper began from a simple observation: an established omnibus
test may fail to use an obvious feature-specific departure
efficiently.  Replacing the primary test is not always necessary.  A
small part of its error budget can instead be reserved for natural
secondary diagnostics.  The main conclusion of the paper is that this
small transfer can substantially broaden power while preserving
almost all of the primary test's original power.

The chain construction is the boundary-selection and accounting
mechanism that makes this augmentation controlled and measurable.
Unconditional budgets specify how much null rejection probability is
retained by the primary statistic and how much is assigned to each
secondary statistic.  Conditional calibration implements those
budgets without inflating the population size.  Once the boundaries
are fixed, the final rule is the corresponding rectangular
intersection.  The ordered first-rejection decomposition measures the
amount of power supplied by the primary test and by each successive
diagnostic.

The normality experiment makes the trade-off concrete.  Only
\(0.75\%\) of the total \(0.05\) error budget was transferred from KS
to variance and skewness.  The resulting loss against normal location
alternatives was very small, while power increased sharply against
several scale, heavy-tailed, and asymmetric alternatives.  Larger
secondary allocations produced further gains for the selected
nonlocation alternatives, but at an increasing cost to the primary
test.  Thus the allocation is a scientific trade-off, not a universal
optimization problem.  The displayed allocation is an illustration,
not an estimate of a universally optimal budget.

The supporting theory establishes that the reusable quantile-based
Monte Carlo calibration converges to the intended population chain.
When an exact finite-\(m\) decision is required for one observation,
the pooled-rank construction provides a separate exchangeability-based
version with exact randomized null size.  These results justify two
different implementations; the reusable null-calibrated chain is the
one studied in the power experiment.

Ordering is primarily an interpretive choice.  The first-rejection
components describe incremental contributions under the chosen order
and should not be read as intrinsic importance measures.  In the
experiment, reversing variance and skewness changed attribution much
more than total power.  The order-invariance proposition gives a
sufficient condition under which the population rejection region is
exactly unchanged, but numerical similarity alone does not establish
that condition.  The weighted minimum-\(p\) comparison provides an
order-free matched-allocation baseline; its hyperrectangle
representation is simply an equivalent description of the same rule.

The numerical study is intentionally focused: it uses one simple null,
one sample size, and secondary statistics with immediate
interpretations, and it is not an exhaustive comparison with
normality tests.  The broader principle is not tied to variance and
skewness.  In other applications, the added diagnostics should encode
the simple departures that matter scientifically and that the primary
test may underweight.  They, their order, and their budgets must be
prespecified or selected by a procedure included in the null
calibration.  Extensions to composite nulls, two-sample problems, and
stages containing several statistics remain useful directions, as do
principled methods for selecting secondary statistics and their small
unconditional budgets.

\subsection*{Declaration of AI-Assisted Content and Workflow}

During the preparation of this manuscript and its supporting materials, the author used a neural-network-based language model (ChatGPT) as an assistive tool. The specific contributions and the author’s supervisory role are as follows:
\begin{enumerate}
\item Language and readability: the assistant was used to refine grammar, improve sentence structure, and enhance the overall clarity of the written text. All technical terminology, tone, and final editorial choices remain the author's own.
\item Code and data processing: the assistant provided support for writing boilerplate iteration loops, debugging minor syntax issues in the R code, and aggregating the large-scale simulation output matrices into the structured LaTeX tables presented in this paper. The core statistical algorithms—including the C++ implementation of the Kolmogorov–Smirnov distance, quantile calibration, and the simulation engine—were written exclusively by the author.
\item Content generation and verification: the AI assistant proposed the addition of the pooled-rank finite-sample version (Section~\ref{sec:exact-pooled-rank}) and assisted in structuring the final rigorous presentation of Proposition \ref{sec:calibration-consistency}’s proof. The author independently evaluated the pooled-rank concept, developed the underlying exchangeability argument, verified every step of the mathematical derivations, and made the final editorial decision to include these sections. The AI did not generate the underlying mathematical strategy.
\end{enumerate}

The paper’s central idea, the proposed augmentation methodology, the experimental design, the interpretation of the power results, and the overall conclusions are exclusively the work of the author. The author assumes full and sole responsibility for the scientific content, mathematical correctness, and final conclusions of this work.

\clearpage
\bibliographystyle{plainurl}
\bibliography{literature_sequential2}{}

\appendix

\section{Proofs of the theoretical results}
\label{appendix_proofs}

\subsection{Proof of Proposition~\ref{prop:order-invariance}}
\label{appendix_proof_order_invariance}

\begin{proof}[Proof of Proposition~\ref{prop:order-invariance}]
Fix a secondary statistic \(T_j\) and a set
\(S\subseteq\{1,\ldots,L\}\setminus\{j\}\).  Conditional null
independence gives
\begin{align}
&Q_0\!\left(
T_j(X)\notin B_j
\,\middle|\,
X\in\mathcal A_0,
\ T_\ell(X)\in B_\ell\text{ for every }\ell\in S
\right)
\nonumber\\
&\qquad=
Q_0\!\left(
T_j(X)\notin B_j
\,\middle|\,
X\in\mathcal A_0
\right)
=
\eta_j,
\label{eq:order-invariance-proof}
\end{align}
whenever the conditioning event has positive probability.  Thus the
region \(B_j\) calibrated in~\eqref{eq:secondary-given-primary}
continues to satisfy its prescribed conditional level after any
subset of the other secondary statistics has been accepted.  By the
assumed uniqueness of the calibration rule, the sequential procedure
therefore selects this same \(B_j\) at whatever position \(T_j\)
occupies.

It follows that the final survivor region is the intersection in
\eqref{eq:order-invariant-regions}.  Since set intersection is
commutative, this region and its complement are identical for all
permutations.  Equation~\eqref{eq:order-invariant-power} follows
because identical rejection regions have identical probability under
every \(Q\).  Finally, applying conditional independence to the
acceptance events gives
\[
Q_0(\mathcal A_L^{\sigma})
=
Q_0(\mathcal A_0)
\prod_{j=1}^L
Q_0\!\left(T_j(X)\in B_j\mid X\in\mathcal A_0\right)
=
(1-\gamma_0)\prod_{j=1}^L(1-\eta_j),
\]
which proves~\eqref{eq:order-invariant-size}.  Although the union of
the stagewise rejection events is unchanged, the event that a
particular statistic is the first to reject changes when the order is
changed; hence its ordered first-rejection contribution can change.
\end{proof}

\subsection{Proof of Proposition~\ref{prop:minp-hyperrectangle}}
\label{appendix_proof_minp}

\begin{proof}[Proof of Proposition~\ref{prop:minp-hyperrectangle}]
By definition of a minimum,
\[
M_w(x)\leq c
\quad\Longleftrightarrow\quad
\frac{p_j(x)}{w_j}\leq c
\text{ for at least one }j\in J_w
\quad\Longleftrightarrow\quad
p_j(x)\leq cw_j
\text{ for at least one }j\in J_w.
\]
The complement is precisely the intersection of the coordinatewise
acceptance conditions in~\eqref{eq:pvalue-hyperrectangle}, proving the
first claim.  For the converse, the choices
in~\eqref{eq:hyperrectangle-to-minp} satisfy
\(w_j\geq0\), \(\sum_jw_j=1\), and \(cw_j=u_j\).  Substitution into
\eqref{eq:minp-hyperrectangle-equivalence} gives exactly the stated
hyperrectangle rule.
\end{proof}

\subsection{Proof of Proposition~\ref{prop:calibration-consistency}}
\label{appendix_proof_consistency}

\begin{proof}[Proof of Proposition~\ref{prop:calibration-consistency}]
Let
\[
Z_i
=
\bigl(T_0(X_i^*),\ldots,T_L(X_i^*)\bigr),
\qquad i=1,\ldots,m,
\]
and let \(Z\) denote the corresponding random vector under \(Q_0\).
The class of axis-aligned rectangles in \(\mathbb R^{L+1}\), allowing
unbounded endpoints, is a Vapnik--Chervonenkis class.  Its pullback
under \(x\mapsto(T_0(x),\ldots,T_L(x))\) is therefore
Glivenko--Cantelli.  On an event of probability one,
\begin{equation}
\label{eq:rectangle-gc}
\sup_{C\in\mathcal C}
\left|\mathbb Q_m(C)-Q_0(C)\right|
\longrightarrow 0,
\end{equation}
where \(\mathcal C\) contains all sets obtained by imposing interval
restrictions on any subset of the statistics.  The uniformity in
\eqref{eq:rectangle-gc} is important because the restrictions defining
\(\widehat{\mathcal A}_{r,m}\) are themselves random.

We proceed by induction over \(r\).  At stage zero, the ordinary
empirical distribution function of \(T_0(X)\) converges uniformly to
its population distribution function.  Consistency of empirical
quantiles at unique continuity points therefore gives convergence of
every endpoint of \(\widehat B_{0,m}\) to the corresponding endpoint
of \(B_0\).  The continuity assumption also implies
\[
Q_0\!\left(
\widehat{\mathcal A}_{0,m}
\mathbin{\triangle}
\mathcal A_0
\right)
\longrightarrow0
\]

Suppose now that the endpoint and symmetric-difference conclusions
hold through stage \(r-1\).  The empirical conditional distribution
function used at stage \(r\) can be written as
\begin{equation}
\label{eq:conditional-empirical-cdf}
\widehat F_{r,m}(t)
=
\frac{
\mathbb Q_m\!\left(
\widehat{\mathcal A}_{r-1,m}
\cap\{x:T_r(x)\le t\}
\right)
}{
\mathbb Q_m(\widehat{\mathcal A}_{r-1,m})
}.
\end{equation}
For the numerator, uniformly in \(t\),
\begin{align*}
&\left|
\mathbb Q_m\!\left(
\widehat{\mathcal A}_{r-1,m}\cap\{T_r\le t\}
\right)
-
Q_0\!\left(
\mathcal A_{r-1}\cap\{T_r\le t\}
\right)
\right|\\
&\qquad\le
\sup_{C\in\mathcal C}|\mathbb Q_m(C)-Q_0(C)|
+
Q_0\!\left(
\widehat{\mathcal A}_{r-1,m}
\mathbin{\triangle}
\mathcal A_{r-1}
\right)
\longrightarrow0.
\end{align*}
The same argument without \(\{T_r\le t\}\) shows that the denominator
of~\eqref{eq:conditional-empirical-cdf} converges to
\(Q_0(\mathcal A_{r-1})>0\).  Hence
\begin{equation}
\label{eq:conditional-cdf-uniform}
\sup_{t\in\mathbb R}
|\widehat F_{r,m}(t)-F_r(t)|
\longrightarrow0.
\end{equation}
Consistency of the required empirical conditional quantiles now gives
endpoint convergence at stage \(r\).  Continuity at those endpoints,
together with the induction hypothesis, yields
\[
Q_0\!\left(
\widehat{\mathcal A}_{r,m}
\mathbin{\triangle}
\mathcal A_r
\right)
\longrightarrow0.
\]
This completes the induction.

Combining the last display with~\eqref{eq:rectangle-gc} gives
\[
\mathbb Q_m(\widehat{\mathcal A}_{r,m})
\longrightarrow Q_0(\mathcal A_r).
\]
Because the population quantile construction has conditional
acceptance probability \(1-\gamma_r\), recursion gives
\[
Q_0(\mathcal A_r)
=
Q_0(\mathcal A_{r-1})(1-\gamma_r)
=
\prod_{s=0}^r(1-\gamma_s),
\]
proving~\eqref{eq:survivor-consistency}.  Moreover,
\[
\widehat\gamma_{r,m}
=
1-
\frac{
\mathbb Q_m(\widehat{\mathcal A}_{r,m})
}{
\mathbb Q_m(\widehat{\mathcal A}_{r-1,m})
}
\longrightarrow
1-
\frac{Q_0(\mathcal A_r)}{Q_0(\mathcal A_{r-1})}
=
\gamma_r,
\]
which proves~\eqref{eq:gammahat-consistency}.  Replacing
\(\mathbb Q_m\) by \(Q_0\) in the same ratio and using the
symmetric-difference convergence proves
\eqref{eq:gammatilde-consistency}.  Finally,
\[
Q_0(\widehat{\mathcal R}^{(L)}_m)
=
1-Q_0(\widehat{\mathcal A}_{L,m})
\longrightarrow
1-Q_0(\mathcal A_L)
=
1-\prod_{r=0}^L(1-\gamma_r),
\]
which is~\eqref{eq:size-consistency}.
\end{proof}

\subsection{Proof of Proposition~\ref{prop:pooled-rank-exact}}
\label{appendix_proof_pooled_rank}

\begin{proof}[Proof of Proposition~\ref{prop:pooled-rank-exact}]
Condition on the integer vector
\((K_{0,M},\ldots,K_{L,M})\).  Under \(H_0\), the pooled sample rows,
together with their iid tie-breaking priorities, are exchangeable.
Permutation equivariance of the complete sequential rule therefore
implies that the indicators
\[
\mathbf 1\{i\in\widetilde D_{r,M}\},
\qquad i=0,\ldots,m,
\]
have equal conditional expectations.  Since exactly \(K_{r,M}\)
indices receive the stage-\(r\) first-rejection label,
\[
\sum_{i=0}^m
\mathbf 1\{i\in\widetilde D_{r,M}\}
=K_{r,M}.
\]
Consequently,
\[
\Pr_0\!\left(
0\in\widetilde D_{r,M}
\mid K_{0,M},\ldots,K_{L,M}
\right)
=
\frac{K_{r,M}}{M}.
\]
Taking expectations and using~\eqref{eq:pooled-count-expectation}
gives
\[
\Pr_0(0\in\widetilde D_{r,M})
=
\frac{E(K_{r,M})}{M}
=b_r,
\]
which proves~\eqref{eq:pooled-exact-stage}.  The sets
\(\widetilde D_{0,M},\ldots,\widetilde D_{L,M}\) are disjoint by the
survivor recursion, so~\eqref{eq:pooled-budget-sum} gives
\[
\Pr_0\!\left(
0\in\bigcup_{r=0}^L\widetilde D_{r,M}
\right)
=
\sum_{r=0}^L b_r
=\alpha.
\]
Finally,
\[
\Pr_0(0\in\widetilde I_{r-1,M})
=
1-\sum_{s=0}^{r-1}b_s
=
\prod_{s=0}^{r-1}(1-\gamma_s).
\]
Because
\(\{0\in\widetilde D_{r,M}\}\subseteq
\{0\in\widetilde I_{r-1,M}\}\),
\eqref{eq:unconditional-stage-budget} implies
\[
\Pr_0\!\left(
0\in\widetilde D_{r,M}
\mid
0\in\widetilde I_{r-1,M}
\right)
=
\frac{b_r}
{\prod_{s=0}^{r-1}(1-\gamma_s)}
=
\gamma_r.
\]
\end{proof}

\section{Finite-\texorpdfstring{\(m\)}{m} uncertainty of calibrated boundaries}
\label{appendix_boundary_uncertainty}

Proposition~\ref{prop:calibration-consistency} establishes convergence
of the estimated boundaries but does not describe their uncertainty at
a fixed calibration size \(m\).  There are two distinct finite-sample
effects.  First, a stage-\(r\) quantile is estimated from only
\[
N_{r,m}:=|I_{r-1,m}|
\]
surviving calibration observations, with
\[
N_{r,m}\approx
m\prod_{s=0}^{r-1}(1-\gamma_s).
\]
Second, for \(r\geq1\), the survivor set itself depends on all
previously estimated boundaries.  Thus uncertainty in an early
boundary propagates to the conditional distribution used at every
later stage.

We first isolate the ordinary order-statistic component.  Let
\(C\subseteq\mathcal X\) be a fixed conditioning region with
\(Q_0(C)>0\), and define
\begin{equation}
\label{eq:fixed-conditional-cdf}
F_{r,C}(t)
:=
Q_0\!\left(T_r(X)\leq t\mid X\in C\right),
\qquad
q_{r,p}(C):=F_{r,C}^{-1}(p).
\end{equation}
Suppose that \(Z_1,\ldots,Z_N\) are iid from \(F_{r,C}\), and write
\(Z_{(1)}\leq\cdots\leq Z_{(N)}\) for their order statistics, with
the conventions \(Z_{(0)}=-\infty\) and
\(Z_{(N+1)}=+\infty\).

\begin{proposition}[Finite-sample interval for a fixed conditional boundary]
\label{prop:boundary-binomial-interval}
Let \(p,\delta\in(0,1)\), and assume that \(F_{r,C}\) is continuous
at its unique \(p\)-quantile.
Let \(B_N\sim\operatorname{Binomial}(N,p)\), and choose integers
\(0\leq\ell\leq u\leq N+1\) satisfying
\begin{equation}
\label{eq:boundary-binomial-coverage}
\Pr\!\left(\ell\leq B_N\leq u-1\right)\geq1-\delta.
\end{equation}
Then
\begin{equation}
\label{eq:boundary-orderstat-interval}
\Pr\!\left(
q_{r,p}(C)\in[Z_{(\ell)},Z_{(u)}]
\right)\geq1-\delta.
\end{equation}
\end{proposition}

\begin{proof}
By continuity and uniqueness,
\(F_{r,C}\bigl(q_{r,p}(C)\bigr)=p\).  Hence the number of observations
strictly below \(q_{r,p}(C)\) has the
\(\operatorname{Binomial}(N,p)\) distribution.  The event
\[
Z_{(\ell)}\leq q_{r,p}(C)\leq Z_{(u)}
\]
is equivalent, up to a null event, to having at least \(\ell\) and
at most \(u-1\) observations below the quantile.  The result follows
from~\eqref{eq:boundary-binomial-coverage}.
\end{proof}

The integers \(\ell\) and \(u\) may be chosen as a central or shortest
binomial acceptance interval.  This construction is
distribution-free and does not require estimation of the density at
the boundary.  For the KS stage one applies it at
\(p=1-\gamma_0\).  For an equal-tailed stage \(r\), it is applied
separately at \(p=\gamma_r/2\) and
\(p=1-\gamma_r/2\).  If \(J\) finite endpoints from fixed
conditional distributions are to be covered simultaneously, using
\(\delta/J\) in place of \(\delta\) for every endpoint gives coverage
at least \(1-\delta\) by the Bonferroni inequality.

The qualification that \(C\) be fixed is important.  In the
implemented sequential calibration,
\(C=\widehat{\mathcal A}_{r-1,m}\) is estimated from the same null
bank that supplies the surviving values.  Therefore
Proposition~\ref{prop:boundary-binomial-interval}, applied naively
with \(N=N_{r,m}\), does not give an unconditional confidence interval
for the population boundary \(q_{r,p}(\mathcal A_{r-1})\).  It
describes the order-statistic component of uncertainty while holding
the upstream filter fixed.  Exact conditional use of the proposition
is possible with an independent stage-specific calibration bank after
the preceding region has been fixed, but such an interval still
targets the boundary conditional on that estimated preceding region
and does not include its upstream estimation error.

For the actual same-bank algorithm, the most direct diagnostic is to
repeat the \emph{entire} sequential calibration independently.  Let
\[
\widehat{\boldsymbol b}^{(z)}_m
=
\left(
\widehat b^{(z)}_{1,m},\ldots,
\widehat b^{(z)}_{J,m}
\right),
\qquad z=1,\ldots,t,
\]
collect all finite lower and upper endpoints produced in repetition
\(z\).  For endpoint \(j\), report its empirical mean, its
calibration-to-calibration standard deviation
\begin{equation}
\label{eq:boundary-repetition-sd}
s_{j,m}
=
\left\{
\frac{1}{t-1}
\sum_{z=1}^t
\left(
\widehat b^{(z)}_{j,m}
-\overline b_{j,m}
\right)^2
\right\}^{1/2},
\end{equation}
and, when useful, the central empirical range
\begin{equation}
\label{eq:boundary-repetition-range}
\left[
\widehat G^{-1}_{j,m}(0.025),
\widehat G^{-1}_{j,m}(0.975)
\right],
\end{equation}
where \(\widehat G_{j,m}\) is the empirical distribution of the
\(t\) replicated boundaries.  These quantities describe the
finite-\(m\) distribution induced by the complete calibration
algorithm and therefore include propagation through all earlier
stages.  The range in~\eqref{eq:boundary-repetition-range} is a
descriptive range of boundary estimates, not an exact confidence
interval for the population boundary.

For this purpose the relevant uncertainty summary is \(s_{j,m}\), not
\(s_{j,m}/\sqrt t\).  The latter measures only the numerical
uncertainty in the estimated \emph{mean} boundary across repetitions
and would understate the variation experienced by a test calibrated
from one finite null bank.  With \(t=100\), standard deviations are
reasonably estimable, whereas the endpoints of a \(95\%\) empirical
range depend on only a few extreme repetitions and should be
interpreted descriptively.  More outer repetitions are appropriate if
stable tail summaries of the boundary distribution are themselves an
inferential target.

\section{Allocation-sensitivity summaries}
\label{appendix_sensitivity}

Table~\ref{tab:sensitivity-rho-full} gives family-average power over
the full \(\rho\) grid under the equal secondary split.  These are
descriptive, unweighted averages over the alternatives included in
each family; in particular, the last two columns do not define a
population-weighted optimal allocation.

\begin{table}[H]
\centering
\small
\setlength{\tabcolsep}{3.2pt}
\begin{threeparttable}
\caption{Family-average power over the full \(\rho\) grid for
\(\omega=0.5\), based on \(m=h=10^6\) and \(t=100\).}
\label{tab:sensitivity-rho-full}
\begin{tabular}{rrrrrrrr}
\toprule
\(\rho\) & Location & Scale & Cauchy & Gamma & Student & Chain, all & Min-\(p\), all \\
\midrule
0.0000 & 0.369113 & 0.459886 & 0.205612 & 0.221051 & 0.083695 & 0.275926 & 0.275926 \\
0.0025 & 0.368897 & 0.725117 & 0.662347 & 0.284837 & 0.383951 & 0.494446 & 0.482419 \\
0.0050 & 0.368688 & 0.741171 & 0.679014 & 0.295637 & 0.395792 & 0.506036 & 0.496918 \\
0.0075 & 0.368473 & 0.749727 & 0.689478 & 0.302989 & 0.403284 & 0.513031 & 0.505347 \\
0.0100 & 0.368261 & 0.755638 & 0.697286 & 0.308680 & 0.409006 & 0.518188 & 0.511011 \\
0.0200 & 0.367394 & 0.769262 & 0.716854 & 0.323557 & 0.423905 & 0.530995 & 0.525231 \\
0.0500 & 0.364757 & 0.787438 & 0.744581 & 0.348803 & 0.446928 & 0.549902 & 0.545063 \\
0.1000 & 0.360215 & 0.801984 & 0.766549 & 0.373022 & 0.467251 & 0.565858 & 0.561784 \\
0.1500 & 0.355441 & 0.810816 & 0.779664 & 0.389873 & 0.480375 & 0.575810 & 0.572395 \\
0.2000 & 0.350449 & 0.817186 & 0.789061 & 0.403307 & 0.490248 & 0.583088 & 0.580231 \\
0.3000 & 0.339506 & 0.826115 & 0.802328 & 0.425012 & 0.504862 & 0.593451 & 0.591581 \\
0.4000 & 0.327074 & 0.832386 & 0.811801 & 0.443007 & 0.515730 & 0.600711 & 0.599705 \\
0.5000 & 0.312708 & 0.837048 & 0.819165 & 0.459362 & 0.524300 & 0.606079 & 0.605675 \\
\bottomrule
\end{tabular}
\begin{tablenotes}[flushleft]\footnotesize
\item Location averages the three normal location alternatives; scale
averages the five normal scale alternatives; the other columns average
the three Cauchy, four standardized gamma, and four Student
alternatives.  ``All'' is the unweighted mean over all 19
alternatives.
\end{tablenotes}
\end{threeparttable}
\end{table}

Table~\ref{tab:sensitivity-omega} fixes the total secondary budget at
\(\rho=0.0075\) and varies only its division between variance and
skewness.  The endpoint \(\omega=1\) omits skewness, whereas
\(\omega=0\) omits variance.

\begin{table}[H]
\centering
\small
\setlength{\tabcolsep}{3.2pt}
\begin{threeparttable}
\caption{Sensitivity to \(\omega\) at \(\rho=0.0075\), based on
\(m=h=10^6\) and \(t=100\).}
\label{tab:sensitivity-omega}
\begin{tabular}{rrrrrrrr}
\toprule
\(\omega\) & Location & Scale & Cauchy & Gamma & Student & Chain, all & Min-\(p\), all \\
\midrule
1.00 & 0.368487 & 0.763748 & 0.669929 & 0.261545 & 0.414516 & 0.507275 & 0.503278 \\
0.75 & 0.368480 & 0.758014 & 0.691785 & 0.297394 & 0.409937 & 0.515799 & 0.509216 \\
0.50 & 0.368473 & 0.749727 & 0.689478 & 0.302989 & 0.403284 & 0.513031 & 0.505347 \\
0.25 & 0.368467 & 0.734752 & 0.681136 & 0.305126 & 0.392547 & 0.505961 & 0.495514 \\
0.00 & 0.368458 & 0.458433 & 0.429507 & 0.291529 & 0.175449 & 0.344946 & 0.339507 \\
\bottomrule
\end{tabular}
\begin{tablenotes}[flushleft]\footnotesize
\item The family averages are defined as in
Table~\ref{tab:sensitivity-rho-full}.  For the chain, \(\omega\) is
the fraction of the secondary first-rejection budget assigned to
variance.  The minimum-\(p\) weights use the same split.
\end{tablenotes}
\end{threeparttable}
\end{table}

Table~\ref{tab:sensitivity-size} reports the minimum and maximum mean
null rejection probabilities over the 13 values of \(\rho\), for
each fixed \(\omega\).  It confirms that the power patterns are not
accompanied by a practically important size change.

\begin{table}[H]
\centering
\small
\begin{threeparttable}
\caption{Ranges of mean null rejection probabilities over the
\(\rho\) grid, based on \(h=10^6\) and \(t=100\).}
\label{tab:sensitivity-size}
\begin{tabular}{rccc}
\toprule
\(\omega\) & KS--Var--Skew & KS--Skew--Var & Weighted min-\(p\) \\
\midrule
1.00 & [0.049975, 0.050003] & [0.049975, 0.050003] & [0.049976, 0.049997] \\
0.75 & [0.049985, 0.050004] & [0.049987, 0.050004] & [0.049984, 0.050002] \\
0.50 & [0.049995, 0.050023] & [0.049995, 0.050023] & [0.049990, 0.050007] \\
0.25 & [0.049995, 0.050036] & [0.049995, 0.050034] & [0.049990, 0.050035] \\
0.00 & [0.049994, 0.050035] & [0.049994, 0.050035] & [0.049990, 0.050031] \\
\bottomrule
\end{tabular}
\begin{tablenotes}[flushleft]\footnotesize
\item Each interval is the range, over the 13 values of \(\rho\), of
the mean null rejection probability across 100 repetitions.  The
Monte Carlo standard errors of the underlying means range from
\(2.83\times10^{-5}\) to \(3.39\times10^{-5}\).
\end{tablenotes}
\end{threeparttable}
\end{table}

\clearpage
\section{Fast vectorized Kolmogorov--Smirnov implementation}
\label{appendix_a}

\begin{lstlisting}
library(Rcpp)
cppFunction("
NumericVector KSDistances(
    NumericMatrix samples, // matrix with samples in rows
    Function F_0 // cdf for H_0
) {
    int m = samples.nrow(); // infer number of samples from samples matrix
    int n = samples.ncol(); // infer sample size from samples matrix

    NumericVector out(m); // allocate memory for the final result

    for (int i = 0; i < m; i++) { // loop over rows of samples matrix
        NumericVector sample = samples(i, _); // get specific sample
        sample.sort(); // sort it in place
        NumericVector F_0_values = F_0(sample); // calculate F_0 for this sample
        double KS_dist_tmp = 0.0; // initiate KS distance for sample
        for (int j = 0; j < n; j++) { // loop over sample order statistics to calc dists
            double t1_KS = double(j + 1) / n - F_0_values[j];
            double t2_KS = F_0_values[j] - double(j) / n;
            double t3_KS = t1_KS > t2_KS ? t1_KS : t2_KS;
            if (t3_KS > KS_dist_tmp) KS_dist_tmp = t3_KS;
        }

        out[i] = KS_dist_tmp; // 1st res col is for KS
    }

    return out;
}
")
\end{lstlisting}

\section{Experiment code}
\label{appendix_b}

\begin{lstlisting}
library(e1071)

h <- 1000000

null.distribution.generation.function <- function(n) rnorm(n)
F_0 <- function(x) pnorm(x)

normal.mean.vector  <- c(0.1, 0.5, 1)
normal.sd.vector    <- c(0.1, 0.2, 1.5, 2, 3)
cauchy.scale.vector <- c(0.1, 0.5, 1)
gamma.par.vector    <- c(0.1, 0.5, 1, 10)
t.df.vector         <- c(1, 2, 3, 4)

alternative.distribution.generation.functions <- c(
    lapply(
        setNames(
            normal.mean.vector,
            sapply(normal.mean.vector, function(s) paste0("mean=", s))
        ),
        function(mean) function(n) rnorm(n, mean, 1)
    ),
    lapply(
        setNames(
            normal.sd.vector,
            sapply(normal.sd.vector, function(s) paste0("sd=", s))
        ),
        function(sd) function(n) rnorm(n, 0, sd)
    ),
    lapply(
        setNames(
            cauchy.scale.vector,
            sapply(cauchy.scale.vector, function(s) paste0("c.s=", s))
        ),
        function(scale) function(n) rcauchy(n, 0, scale)
    ),
    lapply(
        setNames(
            gamma.par.vector,
            sapply(gamma.par.vector, function(s) paste0("g=", s))
        ),
        function(par) function(n) (rgamma(n, shape = par, rate = 1) - par) / sqrt(par)
    ),
    lapply(
        setNames(
            t.df.vector,
            sapply(t.df.vector, function(s) paste0("t.df=", s))
        ),
        function(df) function(n) rt(n, df = df)
    )
)

n <- 10
m <- 1000000
t <- 100

alpha <- 0.05

# rho is the fraction of the total alpha budget assigned to the secondary
# statistics. omega is the fraction of that secondary budget assigned to var.
rho.grid <- c(0, 0.0025, 0.005, 0.0075, 0.01, 0.02, 0.05, 0.10, 0.15, 0.20, 0.3, 0.4, 0.5)
omega.grid <- c(1, 0.75, 0.5, 0.25, 0)


# -----------------------------------------------------------------------------
# Simulation and p-value helpers
# -----------------------------------------------------------------------------

simulate.statistics <- function(distribution.generation.function, number.of.simulations)
{
    samples <- distribution.generation.function(n * number.of.simulations)
    dim(samples) <- c(number.of.simulations, n)

    cbind(
        KS   = KSDistances(samples, F_0),
        var  = apply(samples, 1, var),
        skew = apply(samples, 1, e1071::skewness)
    )
}


p.value.dist.right <- function(Tvals)
{
    rank(-Tvals, ties.method = "max") / length(Tvals)
}


p.value.dist.two.sided.tails <- function(Tvals)
{
    number.of.values <- length(Tvals)

    p.right <- rank(-Tvals, ties.method = "max") / number.of.values
    p.left  <- rank( Tvals, ties.method = "max") / number.of.values

    pmin(1, 2 * pmin(p.left, p.right))
}


# Tnull.sorted must already be sorted. Sorting once per outer repetition avoids
# sorting the same million null statistics for every rho, omega and alternative.
p.value.right.fast <- function(Tnull.sorted, Tobs)
{
    number.of.null.values <- length(Tnull.sorted)

    # Number of null values strictly less than each Tobs
    number.less <- findInterval(Tobs, Tnull.sorted, left.open = TRUE)

    # Number greater than or equal to each Tobs
    number.greater.equal <- number.of.null.values - number.less

    (1 + number.greater.equal) / (number.of.null.values + 1)
}


p.value.two.sided.tails.fast <- function(Tnull.sorted, Tobs)
{
    number.of.null.values <- length(Tnull.sorted)

    # Number of null values <= Tobs
    number.less.equal <- findInterval(Tobs, Tnull.sorted)

    # Number of null values < Tobs
    number.less <- findInterval(Tobs, Tnull.sorted, left.open = TRUE)

    p.left  <- (1 + number.less.equal) / (number.of.null.values + 1)
    p.right <- (1 + number.of.null.values - number.less) /
               (number.of.null.values + 1)

    pmin(1, 2 * pmin(p.left, p.right))
}


add.component.pvalues <- function(stat.matrix, sorted.null.statistics)
{
    cbind(
        stat.matrix,
        KS.pvalue = p.value.right.fast(
            sorted.null.statistics$KS,
            stat.matrix[, "KS"]
        ),
        var.pvalue = p.value.two.sided.tails.fast(
            sorted.null.statistics$var,
            stat.matrix[, "var"]
        ),
        skew.pvalue = p.value.two.sided.tails.fast(
            sorted.null.statistics$skew,
            stat.matrix[, "skew"]
        )
    )
}


weighted.minimum.pvalue <- function(pvalue.matrix, weights)
{
    active.coordinates <- which(weights > 0)

    # A zero weight removes a coordinate. In particular, rho = 0 gives the
    # ordinary KS test instead of a finite empirical bound for var or skew.
    minimum.pvalue <- rep(Inf, nrow(pvalue.matrix))

    for (coordinate in active.coordinates) {
        minimum.pvalue <- pmin(
            minimum.pvalue,
            pvalue.matrix[, coordinate] / weights[coordinate]
        )
    }

    minimum.pvalue
}


two.sided.quantile.bounds <- function(values, conditional.level)
{
    # A zero stage budget means that the stage is inactive. Empirical minima
    # and maxima must not be used here because they would still cause rejection.
    if (conditional.level <= 0) {
        return(c(left = -Inf, right = Inf))
    }

    c(
        left  = unname(quantile(values, conditional.level / 2)),
        right = unname(quantile(values, 1 - conditional.level / 2))
    )
}


err.estimate <- function(check)
    mean(check)


# -----------------------------------------------------------------------------
# Evaluation of all tests for one already-simulated distribution
# -----------------------------------------------------------------------------

check.error.on.distribution <- function(stat.matrix, calibration)
{
    distance.distribution.KS.H1 <- stat.matrix[, "KS"]
    sample.vars.H1              <- stat.matrix[, "var"]
    sample.skews.H1             <- stat.matrix[, "skew"]

    KS.1stage.check  <- distance.distribution.KS.H1 > calibration$KS.bound.1stage
    KS.2stages.check <- distance.distribution.KS.H1 > calibration$KS.bound.2stages
    KS.3stages.check <- distance.distribution.KS.H1 > calibration$KS.bound.3stages

    var.rb.1stage.check <- sample.vars.H1 > calibration$var.rb.1stage
    var.lb.1stage.check <- sample.vars.H1 < calibration$var.lb.1stage

    rb.sample.var.check.q.stage.1.2stages <-
        sample.vars.H1 > calibration$rb.sample.var.q.stage.1.2stages
    lb.sample.var.check.q.stage.1.2stages <-
        sample.vars.H1 < calibration$lb.sample.var.q.stage.1.2stages
    rb.sample.var.check.q.stage.1.3stages <-
        sample.vars.H1 > calibration$rb.sample.var.q.stage.1.3stages
    lb.sample.var.check.q.stage.1.3stages <-
        sample.vars.H1 < calibration$lb.sample.var.q.stage.1.3stages
    rb.sample.var.check.q.stage.2.3stages <-
        sample.vars.H1 > calibration$rb.sample.var.q.stage.2.3stages
    lb.sample.var.check.q.stage.2.3stages <-
        sample.vars.H1 < calibration$lb.sample.var.q.stage.2.3stages

    skew.rb.1stage.check <- sample.skews.H1 > calibration$skew.rb.1stage
    skew.lb.1stage.check <- sample.skews.H1 < calibration$skew.lb.1stage

    rb.sample.skew.check.q.stage.1.2stages <-
        sample.skews.H1 > calibration$rb.sample.skew.q.stage.1.2stages
    lb.sample.skew.check.q.stage.1.2stages <-
        sample.skews.H1 < calibration$lb.sample.skew.q.stage.1.2stages
    rb.sample.skew.check.q.stage.1.3stages <-
        sample.skews.H1 > calibration$rb.sample.skew.q.stage.1.3stages
    lb.sample.skew.check.q.stage.1.3stages <-
        sample.skews.H1 < calibration$lb.sample.skew.q.stage.1.3stages
    rb.sample.skew.check.q.stage.2.3stages <-
        sample.skews.H1 > calibration$rb.sample.skew.q.stage.2.3stages
    lb.sample.skew.check.q.stage.2.3stages <-
        sample.skews.H1 < calibration$lb.sample.skew.q.stage.2.3stages

    minimum.pvalue <- weighted.minimum.pvalue(
        stat.matrix[, c("KS.pvalue", "var.pvalue", "skew.pvalue"), drop = FALSE],
        calibration$minimum.pvalue.weights
    )
    minimum.pvalue.reject <- minimum.pvalue < calibration$pvalue.min.bound


    # -------------------------------------------------------------------------
    # Raw secondary rejection indicators
    # -------------------------------------------------------------------------

    var.reject.stage.1.2stages <-
        rb.sample.var.check.q.stage.1.2stages |
        lb.sample.var.check.q.stage.1.2stages

    skew.reject.stage.1.2stages <-
        rb.sample.skew.check.q.stage.1.2stages |
        lb.sample.skew.check.q.stage.1.2stages

    var.reject.stage.1.3stages <-
        rb.sample.var.check.q.stage.1.3stages |
        lb.sample.var.check.q.stage.1.3stages

    skew.reject.stage.2.3stages <-
        rb.sample.skew.check.q.stage.2.3stages |
        lb.sample.skew.check.q.stage.2.3stages

    skew.reject.stage.1.3stages <-
        rb.sample.skew.check.q.stage.1.3stages |
        lb.sample.skew.check.q.stage.1.3stages

    var.reject.stage.2.3stages <-
        rb.sample.var.check.q.stage.2.3stages |
        lb.sample.var.check.q.stage.2.3stages


    # -------------------------------------------------------------------------
    # Decomposition: KS -> var
    # -------------------------------------------------------------------------

    R0.KS.var <- KS.2stages.check
    R1.KS.var <- !R0.KS.var & var.reject.stage.1.2stages
    reject.KS.var <- R0.KS.var | R1.KS.var


    # -------------------------------------------------------------------------
    # Decomposition: KS -> skew
    # -------------------------------------------------------------------------

    R0.KS.skew <- KS.2stages.check
    R1.KS.skew <- !R0.KS.skew & skew.reject.stage.1.2stages
    reject.KS.skew <- R0.KS.skew | R1.KS.skew


    # -------------------------------------------------------------------------
    # Decomposition: KS -> var -> skew
    # -------------------------------------------------------------------------

    R0.KS.var.skew <- KS.3stages.check
    R1.KS.var.skew <-
        !R0.KS.var.skew & var.reject.stage.1.3stages
    R2.KS.var.skew <-
        !R0.KS.var.skew &
        !var.reject.stage.1.3stages &
        skew.reject.stage.2.3stages
    reject.KS.var.skew <-
        R0.KS.var.skew | R1.KS.var.skew | R2.KS.var.skew


    # -------------------------------------------------------------------------
    # Decomposition: KS -> skew -> var
    # -------------------------------------------------------------------------

    R0.KS.skew.var <- KS.3stages.check
    R1.KS.skew.var <-
        !R0.KS.skew.var & skew.reject.stage.1.3stages
    R2.KS.skew.var <-
        !R0.KS.skew.var &
        !skew.reject.stage.1.3stages &
        var.reject.stage.2.3stages
    reject.KS.skew.var <-
        R0.KS.skew.var | R1.KS.skew.var | R2.KS.skew.var


    # -------------------------------------------------------------------------
    # Return total power and its decomposition
    # -------------------------------------------------------------------------

    c(
        KS.total = err.estimate(KS.1stage.check),
        var.total = err.estimate(var.lb.1stage.check | var.rb.1stage.check),
        skew.total = err.estimate(skew.lb.1stage.check | skew.rb.1stage.check),

        minimum.pvalue.total = err.estimate(minimum.pvalue.reject),

        KS.var.total = err.estimate(reject.KS.var),
        KS.var.KS.contribution = err.estimate(R0.KS.var),
        KS.var.var.contribution = err.estimate(R1.KS.var),

        KS.skew.total = err.estimate(reject.KS.skew),
        KS.skew.KS.contribution = err.estimate(R0.KS.skew),
        KS.skew.skew.contribution = err.estimate(R1.KS.skew),

        KS.var.skew.total = err.estimate(reject.KS.var.skew),
        KS.var.skew.KS.contribution = err.estimate(R0.KS.var.skew),
        KS.var.skew.var.contribution = err.estimate(R1.KS.var.skew),
        KS.var.skew.skew.contribution = err.estimate(R2.KS.var.skew),

        KS.skew.var.total = err.estimate(reject.KS.skew.var),
        KS.skew.var.KS.contribution = err.estimate(R0.KS.skew.var),
        KS.skew.var.skew.contribution = err.estimate(R1.KS.skew.var),
        KS.skew.var.var.contribution = err.estimate(R2.KS.skew.var)
    )
}


test.names <- c(
    "KS.total",
    "var.total",
    "skew.total",
    "minimum.pvalue.total",
    "KS.var.total",
    "KS.var.KS.contribution",
    "KS.var.var.contribution",
    "KS.skew.total",
    "KS.skew.KS.contribution",
    "KS.skew.skew.contribution",
    "KS.var.skew.total",
    "KS.var.skew.KS.contribution",
    "KS.var.skew.var.contribution",
    "KS.var.skew.skew.contribution",
    "KS.skew.var.total",
    "KS.skew.var.KS.contribution",
    "KS.skew.var.skew.contribution",
    "KS.skew.var.var.contribution"
)


rho.names <- paste0("rho=", format(rho.grid, trim = TRUE, scientific = FALSE))
omega.names <- paste0("omega=", format(omega.grid, trim = TRUE, scientific = FALSE))

power.estimates <- array(
    NA_real_,
    dim = c(
        length(test.names),
        length(alternative.distribution.generation.functions),
        length(rho.grid),
        length(omega.grid),
        t
    ),
    dimnames = list(
        test = test.names,
        alternative = names(alternative.distribution.generation.functions),
        rho = rho.names,
        omega = omega.names,
        repetition = seq_len(t)
    )
)

type1.error.estimates <- array(
    NA_real_,
    dim = c(length(test.names), length(rho.grid), length(omega.grid), t),
    dimnames = list(
        test = test.names,
        rho = rho.names,
        omega = omega.names,
        repetition = seq_len(t)
    )
)


# -----------------------------------------------------------------------------
# Repeated calibration and paired sensitivity study
# -----------------------------------------------------------------------------

for (z in seq_len(t))
{
    # One H0 calibration bank is shared across all rho and omega values in this
    # outer repetition.

    print("started")
    print(z)
    flush.console()
    
    stats.H0.cal <- simulate.statistics(
        null.distribution.generation.function,
        m
    )

    print("H0 samples simulation finished")
    flush.console()

    distance.distribution.KS <- stats.H0.cal[, "KS"]
    sample.vars              <- stats.H0.cal[, "var"]
    sample.skews             <- stats.H0.cal[, "skew"]

    print("H0 statitics computation finished")
    flush.console()

    sorted.null.statistics <- list(
        KS   = sort(distance.distribution.KS),
        var  = sort(sample.vars),
        skew = sort(sample.skews)
    )

    null.component.pvalues <- cbind(
        KS.pvalue   = p.value.dist.right(distance.distribution.KS),
        var.pvalue  = p.value.dist.two.sided.tails(sample.vars),
        skew.pvalue = p.value.dist.two.sided.tails(sample.skews)
    )

    # Ordinary one-stage bounds do not depend on rho or omega.
    KS.bound.1stage <- unname(quantile(distance.distribution.KS, 1 - alpha))

    var.bounds.1stage <- two.sided.quantile.bounds(sample.vars, alpha)
    var.lb.1stage <- var.bounds.1stage["left"]
    var.rb.1stage <- var.bounds.1stage["right"]

    skew.bounds.1stage <- two.sided.quantile.bounds(sample.skews, alpha)
    skew.lb.1stage <- skew.bounds.1stage["left"]
    skew.rb.1stage <- skew.bounds.1stage["right"]

    calibration.grid <- vector("list", length(rho.grid))

    for (rho.index in seq_along(rho.grid))
    {
        rho <- rho.grid[rho.index]
        calibration.grid[[rho.index]] <- vector("list", length(omega.grid))

        # Unconditional stage budgets. At rho = 0.0075 and omega = 0.5,
        # these are 0.049625, 0.0001875 and 0.0001875.
        stage.0.budget <- alpha * (1 - rho)
        secondary.stage.budget <- alpha * rho

        # Two-stage KS -> var and KS -> skew tests give the whole secondary
        # budget to their single secondary statistic.
        gamma.0.2stages <- stage.0.budget
        gamma.1.2stages <- secondary.stage.budget /
                           (1 - gamma.0.2stages)

        KS.bound.2stages <- unname(quantile(
            distance.distribution.KS,
            1 - gamma.0.2stages
        ))

        # The stage-0 budget is the same for the two- and three-stage tests.
        KS.bound.3stages <- KS.bound.2stages

        filter.stage.0.2stages <-
            distance.distribution.KS <= KS.bound.2stages
        filter.stage.0.3stages <- filter.stage.0.2stages

        sample.vars.filtered.by.stage.0.2stages <-
            sample.vars[filter.stage.0.2stages]
        var.bounds.stage.1.2stages <- two.sided.quantile.bounds(
            sample.vars.filtered.by.stage.0.2stages,
            gamma.1.2stages
        )

        sample.skews.filtered.by.stage.0.2stages <-
            sample.skews[filter.stage.0.2stages]
        skew.bounds.stage.1.2stages <- two.sided.quantile.bounds(
            sample.skews.filtered.by.stage.0.2stages,
            gamma.1.2stages
        )

        for (omega.index in seq_along(omega.grid))
        {
            omega <- omega.grid[omega.index]

            # var and skew budgets sum to the total secondary-stage budget.
            var.stage.budget  <- alpha * rho * omega
            skew.stage.budget <- alpha * rho * (1 - omega)

            # KS -> var -> skew. Each gamma is a conditional level obtained
            # from the corresponding unconditional stage budget.
            gamma.0.3stages <- stage.0.budget
            gamma.var.stage.1.3stages <-
                var.stage.budget / (1 - gamma.0.3stages)
            gamma.skew.stage.2.3stages <-
                skew.stage.budget /
                ((1 - gamma.0.3stages) *
                 (1 - gamma.var.stage.1.3stages))

            # KS -> skew -> var uses the same unconditional budgets, with the
            # order of the secondary stages reversed.
            gamma.skew.stage.1.3stages <-
                skew.stage.budget / (1 - gamma.0.3stages)
            gamma.var.stage.2.3stages <-
                var.stage.budget /
                ((1 - gamma.0.3stages) *
                 (1 - gamma.skew.stage.1.3stages))

            # Stage 1: KS -> var -> skew
            sample.vars.filtered.by.stage.0.3stages <-
                sample.vars[filter.stage.0.3stages]
            var.bounds.stage.1.3stages <- two.sided.quantile.bounds(
                sample.vars.filtered.by.stage.0.3stages,
                gamma.var.stage.1.3stages
            )

            filter.var.stage.1.3stages <-
                sample.vars >= var.bounds.stage.1.3stages["left"] &
                sample.vars <= var.bounds.stage.1.3stages["right"]

            # Stage 2: KS -> var -> skew
            sample.skews.filtered.by.stage.0.1.3stages <-
                sample.skews[
                    filter.stage.0.3stages &
                    filter.var.stage.1.3stages
                ]
            skew.bounds.stage.2.3stages <- two.sided.quantile.bounds(
                sample.skews.filtered.by.stage.0.1.3stages,
                gamma.skew.stage.2.3stages
            )

            # Stage 1: KS -> skew -> var
            sample.skews.filtered.by.stage.0.3stages <-
                sample.skews[filter.stage.0.3stages]
            skew.bounds.stage.1.3stages <- two.sided.quantile.bounds(
                sample.skews.filtered.by.stage.0.3stages,
                gamma.skew.stage.1.3stages
            )

            filter.skew.stage.1.3stages <-
                sample.skews >= skew.bounds.stage.1.3stages["left"] &
                sample.skews <= skew.bounds.stage.1.3stages["right"]

            # Stage 2: KS -> skew -> var
            sample.vars.filtered.by.stage.0.1.3stages <-
                sample.vars[
                    filter.stage.0.3stages &
                    filter.skew.stage.1.3stages
                ]
            var.bounds.stage.2.3stages <- two.sided.quantile.bounds(
                sample.vars.filtered.by.stage.0.1.3stages,
                gamma.var.stage.2.3stages
            )

            # Weighted minimum-p test. These weights agree with the same
            # allocation used by the sequential tests.
            minimum.pvalue.weights <- c(
                KS   = 1 - rho,
                var  = rho * omega,
                skew = rho * (1 - omega)
            )

            null.minimum.pvalues <- weighted.minimum.pvalue(
                null.component.pvalues,
                minimum.pvalue.weights
            )
            pvalue.min.bound <- unname(quantile(
                null.minimum.pvalues,
                alpha
            ))

            calibration.grid[[rho.index]][[omega.index]] <- list(
                rho = rho,
                omega = omega,
                stage.0.budget = stage.0.budget,
                var.stage.budget = var.stage.budget,
                skew.stage.budget = skew.stage.budget,
                gamma.0.2stages = gamma.0.2stages,
                gamma.1.2stages = gamma.1.2stages,
                gamma.0.3stages = gamma.0.3stages,
                gamma.var.stage.1.3stages = gamma.var.stage.1.3stages,
                gamma.skew.stage.2.3stages = gamma.skew.stage.2.3stages,
                gamma.skew.stage.1.3stages = gamma.skew.stage.1.3stages,
                gamma.var.stage.2.3stages = gamma.var.stage.2.3stages,
                KS.bound.1stage = KS.bound.1stage,
                KS.bound.2stages = KS.bound.2stages,
                KS.bound.3stages = KS.bound.3stages,
                var.rb.1stage = var.rb.1stage,
                var.lb.1stage = var.lb.1stage,
                skew.rb.1stage = skew.rb.1stage,
                skew.lb.1stage = skew.lb.1stage,
                rb.sample.var.q.stage.1.2stages =
                    var.bounds.stage.1.2stages["right"],
                lb.sample.var.q.stage.1.2stages =
                    var.bounds.stage.1.2stages["left"],
                rb.sample.var.q.stage.1.3stages =
                    var.bounds.stage.1.3stages["right"],
                lb.sample.var.q.stage.1.3stages =
                    var.bounds.stage.1.3stages["left"],
                rb.sample.skew.q.stage.1.2stages =
                    skew.bounds.stage.1.2stages["right"],
                lb.sample.skew.q.stage.1.2stages =
                    skew.bounds.stage.1.2stages["left"],
                rb.sample.skew.q.stage.1.3stages =
                    skew.bounds.stage.1.3stages["right"],
                lb.sample.skew.q.stage.1.3stages =
                    skew.bounds.stage.1.3stages["left"],
                rb.sample.skew.q.stage.2.3stages =
                    skew.bounds.stage.2.3stages["right"],
                lb.sample.skew.q.stage.2.3stages =
                    skew.bounds.stage.2.3stages["left"],
                rb.sample.var.q.stage.2.3stages =
                    var.bounds.stage.2.3stages["right"],
                lb.sample.var.q.stage.2.3stages =
                    var.bounds.stage.2.3stages["left"],
                minimum.pvalue.weights = minimum.pvalue.weights,
                pvalue.min.bound = pvalue.min.bound
            )
        }
    }

    print("omega rho grid, preliminary computations done")
    flush.console()

    # Independent H0 validation bank. It is shared across all grid points so
    # the type-I-error comparisons are paired.
    stats.H0.val <- add.component.pvalues(
        simulate.statistics(null.distribution.generation.function, h),
        sorted.null.statistics
    )

    for (rho.index in seq_along(rho.grid)) {
        for (omega.index in seq_along(omega.grid)) {
            type1.error.estimates[, rho.index, omega.index, z] <-
                check.error.on.distribution(
                    stats.H0.val,
                    calibration.grid[[rho.index]][[omega.index]]
                )
        }
    }

    print("type 1 error done")
    flush.console()

    # Each H1 bank is generated once and then reused across every rho and omega.
    # Processing alternatives one at a time avoids keeping all H1 banks in RAM.
    for (alternative.index in
         seq_along(alternative.distribution.generation.functions))
    {
        stats.H1 <- add.component.pvalues(
            simulate.statistics(
                alternative.distribution.generation.functions[[alternative.index]],
                h
            ),
            sorted.null.statistics
        )

        for (rho.index in seq_along(rho.grid)) {
            for (omega.index in seq_along(omega.grid)) {
                power.estimates[, alternative.index, rho.index, omega.index, z] <-
                    check.error.on.distribution(
                        stats.H1,
                        calibration.grid[[rho.index]][[omega.index]]
                    )
            }
        }

        print("alternative done")
        print(alternative.index)
        flush.console()
    }

    print(z)
    flush.console()
}


# -----------------------------------------------------------------------------
# Summaries
# -----------------------------------------------------------------------------

power.mean <- apply(power.estimates, c(1, 2, 3, 4), mean)
power.sd   <- apply(power.estimates, c(1, 2, 3, 4), sd)
power.se   <- power.sd / sqrt(t)

type1.error.mean <- apply(type1.error.estimates, c(1, 2, 3), mean)
type1.error.sd   <- apply(type1.error.estimates, c(1, 2, 3), sd)
type1.error.se   <- type1.error.sd / sqrt(t)

sensitivity.results <- list(
    rho.grid = rho.grid,
    omega.grid = omega.grid,
    power.estimates = power.estimates,
    power.mean = power.mean,
    power.sd = power.sd,
    power.se = power.se,
    type1.error.estimates = type1.error.estimates,
    type1.error.mean = type1.error.mean,
    type1.error.sd = type1.error.sd,
    type1.error.se = type1.error.se
)
\end{lstlisting}

\end{document}